\documentclass[12pt]{article}

\usepackage{palatino}
\usepackage{graphicx}
\usepackage{geometry}
\usepackage{hyperref}
\usepackage{soul,color}

\usepackage{pstricks}
\usepackage{amsthm,amsfonts,amsmath,amssymb,mathtools,graphics}
\usepackage{enumerate}
\usepackage{enumitem}
\usepackage{paralist}
\usepackage{color}

\usepackage[stretch=10]{microtype}
\usepackage{mathrsfs}
\usepackage{esint}
\usepackage{natbib}
\newcommand{\xnorm}[1]{ \Vert #1 \Vert }
\newcommand{\var}{\text{\rm var}}
\newcommand{\sgn}{\text{\rm sgn}}
\newcommand{\zz}[1]{\mathbb{#1}}

\makeatletter

\usepackage[small,hang]{caption}\usepackage{latexsym}\usepackage{bm}\usepackage{mathrsfs}\usepackage{alltt}\usepackage{eucal}\usepackage{url}\usepackage{colortbl}\usepackage{fancybox}\usepackage{array}\usepackage{rotating}\usepackage{multirow}\usepackage{longtable}\usepackage{multicol}\usepackage{xcolor}\usepackage{eurosym}\usepackage{booktabs}\usepackage{multimedia}\usepackage{epigraph}\usepackage{bigints}\usepackage{tabularx}\usepackage{mathrsfs}\usepackage{rotating}\usepackage{lscape}

\makeatother

\newtheorem{theorem}{Theorem}
\newtheorem{proposition}{Proposition}
\newtheorem{lemma}[proposition]{Lemma}
\newtheorem{corollary}[proposition]{Corollary}

\newtheorem{claim}[proposition]{Claim}

\newtheorem{berryEsseenTh}{Berry-Esseen Theorem}

\newtheorem{hoeffding}{Hoeffding's Inequality}

\newcommand{\xset}[1]{\left\{ #1 \right\}}

\begin{document}

\title{Nash Equilibria for Quadratic Voting\thanks{We are grateful to numerous colleagues for useful comments.  We especially appreciate the role of Lars Peter Hansen for introducing us and suggesting our collaboration.}}

\author{Steven P. Lalley%
\thanks{Department of Statistics, University of Chicago: 5734 S. University
Avenue, Chicago, IL 60637; lalley@galton.uchicago.edu, http://www.stat.uchicago.edu/lalley/.%
} \and E.~Glen Weyl%
\thanks{Microsoft Research New England, 641 Avenue of the Americas, New York, NY 10011; glenweyl@microsoft.com,
http://www.glenweyl.com.%
}}

\date{July 2019}

\maketitle

\thispagestyle{empty}

\begin{abstract}
  A group of $N$ individuals must choose between two collective
  alternatives.  Under Quadratic Voting (QV), agents buy votes in
  favor of their preferred alternative from a clearing house, paying
  the square of the number of votes purchased; the sum of all votes
  purchased determines the outcome. We provide the first rigorous
  results for this mechanism, in a canonical independent private
  values environment with bounded value distributions.  In addition to
  characterizing the nature of equilibria, we demonstrate that for all
  bounded value distributions, the utilitarian welfare losses of the
  mechanism as a proportion of the maximum possible welfare tends to
  zero as the population size becomes large.
\end{abstract}


\section{Introduction}\label{sec:intro}

Consider a binary collective-decision problem in which a group of $N$
individuals must choose between two alternatives. Each individual $i$
has a privately known value $u_{i}$ that determines her willingness to
pay for one alternative over the other; positive values indicate
affinity for outcome $+1$, negative values for outcome $-1$.  {\em
Quadratic Voting} ({\em QV}) is a simple and detail-free mechanism
designed to maximize utilitarian efficiency in this
setting.\footnote{Clearly many other objectives are possible for this
problem, and many involve distributional considerations. However, we
focus on a utilitarian objective because it is the one most
extensively studied in the literature \citep{bowen, groves}.} In this
system, individuals buy votes (either negative or positive, depending
on which alternative is favored) from a clearing house, paying the
square of the number of votes purchased.  The sum of all votes
purchased then determines the outcome. The utility (payoff) of the
outcome to an individual with value $u$ is $+u$ if outcome $+$ is
adopted, but $-u$ if outcome $-$ is adopted.\footnote{Our results will
apply to a modified version of the problem in which the utility is
``smoothed'' in such a way that each voter's utility is a continuous
function of the vote total. See section~\ref{ssec:model} for details.}

QV is of interest for several reasons.  First, it is among a small
collection of mechanisms whose approximate efficiency has been
suggested to hold under relatively broad conditions
\citep{tidemansurvey}.  Second, several non-rigorous calculations
based on asymptotic approximation \citep{robustnessQV} and numerical
work \citep{finitepopulations}, as well as both laboratory
\citep{storableqv} and field experiments \citep{wild, holland,
  cavaille} provide some evidence that QV may be approximately
efficient --- or at least of practical utility --- in a wider range of
environments than other approximately efficient mechanisms.  Third,
variations on QV have recently been implemented in a wide range of
field settings with stakes of hundreds of millions of dollars
\citep{qvimplement, colorado}.   Prior to the first draft of this
paper (published on arxiv in 2014), there were no fully rigorous
game-theoretic proofs analyses of the mechanism (though some have
emerged since the first draft). It is the aim of this article to
present such a rigorous analysis.

The heuristic rationale for QV is quite simple. The marginal benefit
to a voter of an additional vote is her value multiplied by her {\em
marginal pivotality} (roughly, the perceived probability that an
additional unit of vote will sway the decision).  She maximizes
utility by equating this marginal utility to the linear marginal cost of a
vote. Therefore, \emph{if voters share the same marginal pivotality}, they
will buy votes in proportion to their values, thus bringing
about utilitarian efficiency.  Furthermore, the quadratic cost
function is the \emph{unique} cost function with this property. This
argument is explained in further detail  in Subsection
\ref{intuition}.

Variations of this rationale have been used to justify quadratic
mechanisms in a number of related collective-decision-making
problems. However, to our knowledge, this heuristic rationale has
never been translated into a rigorous argument for efficiency in the
sort of non-cooperative, incomplete information game theoretic model
in which mechanisms for allocating private goods have been studied, at
least since the work of \citet{myersonauction}.  In fact, as we
will show, in the modified setting of quadratic voting that we will
consider the crucial \emph{ansatz} of the rationale -- that in
equilibrium all voters will have the same marginal pivotality -- is
false. As a result, formal equilibrium analysis is a far more subtle
task than the heuristic argument of the previous paragraph might
suggest.  Nevertheless, we will show that for voters with values in
the ``bulk'' of the distribution, the marginal pivotality is
approximately constant.\footnote{Theorems of \citet{kahn}
and \citet{nabil} imply marginal pivotality must
converge to zero as $N \rightarrow \infty$. Our results will show
that, with probability approaching $1$, the \emph{ratio} between the
marginal pivotalities of two randomly chosen voters will be close to
$1$.}

While interest in QV has grown in recent years, the core idea behind
it is one of the oldest in the mathematical theory of
voting. \cite{penrose} observed that the, if voters act randomly, the
power of an individual with $m$ votes relative to that of an
individual with a single vote grows as $m^2$.  This suggests that the
voting weight of sub-units in a system of multi-level representation
should be proprotional to the square root of population size, an
option considered seriously by the European Union \citep{compromise}.
QV, as we analyze, may be seen as an application of this intuition to
the setting where agent can express preference intensity using a
scarce resource. 

Despite this connection, the use of quadratic pricing for collective decision-making was
apparently first suggested by \citet{gl}, who proposed it as a Nash
implementation of the optimal level of continuous public goods under
complete information that avoids the fragility of previously suggested
efficient mechanisms.  \citet{hyzeck} provided the first variant of
the heuristic rationale above (based on incentive theory, rather than chances of being pivotal) to uniquely justify quadratic pricing
mechanism and proposed an iterative procedure that they conjectured
would converge to \citet{gl}'s complete information optimum in the
presence of private information. In a preprint circulated in 2012,
\citet{qvb} first proposed the use of QV for binary collective
decision problems, and conjectured that it would lead to
asymptotically efficient decisions in the environment considered in
the present paper. (This was based on an extension of \citeauthor{hyzeck}'s
heuristic rationale.)  \citet{goeree} independently suggested using a
detail-based, approximately direct variant of QV in the special case
where values are sampled from zero-mean normal distributions, and
derived an equilibrium in the case $N=2$.\footnote{\citeauthor{goeree}
  also derive an asymptotic efficiency result (cf. Proposition 2 of
  their article).  While the result is
  true, our analysis below shows that their short proof is
  fallacious. The flaw is  their use of the central limit
  theorem, in particular, the assertion
  \begin{quote}
    for a large electorate, the central limit theorem implies that
    $G(b_{i})$ limits to $\Phi (b_{i}/\sqrt{2 \pi n }\sigma)$.
  \end{quote}
  This assertion is problematic for at least two critical reasons: (1)
  Because the optimal ``bid'' rule $v_{i}\mapsto b_{i}$ is \emph{a
    priori} unknown -- and in principle dependent on the size $N$ of the
  electorate -- the distribution of the summands $b_{j}$ could in
  principle vary wildly as $N \rightarrow\infty$. (2) Even in the
  simplest settings where a central limit theorem holds (e.g., the
  DeMoivre-Laplace theorem for sums of Bernoulli random variables) the
  normal approximation breaks down in the tails of the
  distribution. This  is, as we will show later,
  precisely the region where voters  might take an ``extremist''
  position, producing large discontinuities in the function
  $v_{i}\mapsto b_{i}$. }

Our main results hold generally in the non-cooperative, independent
private values setting with arbitrary value distributions whose
supports are bounded and which have smooth densities that do not
vanish at the endpoints of the support. This range of applicability is
broader than is common in the literature for normative analysis of
non-direct mechanisms; for example, \citet{casella} studies only a few
specific value distributions numerically for fixed population sizes.
However, our model assumptions are almost certainly too narrow to
encompass the broad range of settings where QV is already being
applied practically, which are likely to involve complications such as
aggregate uncertainty and collusion.  Results of \citet{myatt} and
\citet{robustnessQV} suggest that asymptotic efficiency fails
generically with aggregate uncertainty, as different agents have
different estimates of their marginal pivotalities based on their private
information \citep{robustnessQV}.  Similarly, \citet{finitepopulations}
show numerical examples of imperfect efficiency in finite populations.
To our knowledge no example has been shown where these lead to a
welfare loss of greater than 10\%, in contrast to (for example)
standard one-person-one-vote which often loses 100\% of potential
welfare.

Despite its somewhat limited scope of applicability, our analysis
makes several important theoretical contributions to the literature
surrounding QV. First, we develop techniques to account for deviations
from the conjectured linear asymptotic equilibrium, using a series of
increasingly precise approximations that allow us to successively
bound these deviations and their impact on incentives in a
``ratchet''.  Second, we show that at least in the ``unbalanced'' case
of a value distribution with non-zero mean, very substantial
deviations from linearity do in fact occur in the extreme tails of the
distribution, and that these deviations actually ``drive'' the
equilibrium, by forcing agents in the bulk of distribution to buy
votes as insurance against ``extremists''. Our techniques have
already  been applied to study the performance of non-quadratic
rules \citep{eguia} and to establish analogous approximate efficiency
results for a version of QV for social choices with more than two
alternatives \citep{multioption}, both under complete information.

In summary, this paper makes the following three main
contributions to the study of quadratic voting:

(1)  It provides the first fully rigorous equilibrium analysis for
QV.

(2) It is  the only fully rigorous equilibrium analysis that allows
for incomplete information.

(3)  It shows that significant deviations from linearity do in fact
  occur in equilibrium, and at the same time introduces a suite of
  mathematical techniques for dealing with these nonlinearities.

\section{Statement of Main Results}

\subsection{Model Assumptions}\label{ssec:model}

We consider an independent symmetric private values environment with
$N$ voters $i=1, \ldots, N$.  Each voter $i$ is characterized by a
value, $u_i$; these values are drawn independently from a continuous
probability distribution $F$ with $C^{\infty}$, strictly
positive density $f$ supported by a finite closed interval
$\left[\underline u, \overline u\right]$.\footnote{The
assumption that the density $f$ is positive at the endpoints
$\underline u ,\overline u$ is of critical importance for our main
results, as ``extremists'' play a crucial role in the Bayes-Nash
equilibria for the game. Our methods would extend to densities $f$ that
vanish at one or both of the endpoints, but the nature of the
Bayes-Nash equilibria changes in these cases.} Each individual knows
her own value, but not the values of any of the other $N-1$ voters;
however, the sampling distribution $F$ is known to all. For
normalization, we assume the numeraire has been scaled so that $\min
(-\underline u, \overline u)\geq 1$.  We denote by $\mu$, $\sigma^2$,
and $\mu_3$, respectively, respectively the mean, variance, and raw
third moments of $u$ under $F$.

We consider a variant of the payoff described above, in which the
utility of the outcome is ``smoothed''.  \footnote{Although both the
discrete binary choice set-up of \cite{qvb} and the continuous public
goods model of \cite{hyzeck} helped inspire this model, it differs
from both. Consequently, our results have no direct implications for
those models. It differs from \cite{qvb}'s model in that the outcome
is smoothed rather than jumping discontinuously at $0$.  It differs in
a variety of respects from \cite{hyzeck}'s, notably in that utility is
linear in the common and bounded outcome, whereas \cite{hyzeck} assume
strictly concave preferences with heterogeneous ideal points and an
outcome that may take values in the full real space.  \cite{hyzeck}
also consider a multidimensional issue space with no access to
transfers and an iterative procedure to converge to this outcome, none
of which feature in our model.}  Each voter $i$ chooses a number of
votes $v_i\in \zz R$ to buy, and pays $v_i^2$ dollars for these. The
payoff to voter $i$ is then
\begin{equation}\label{eq:payoff}
	\Psi (V)u_{i}, \quad \text{where} \; V=\sum_{i=1}^{N} v_{i}
\end{equation}
is the vote total and $\Psi :\zz{R}\rightarrow [-1,1]$ is an odd, nondecreasing,
$C^{\infty}$ function such that for some $\delta >0$,
\begin{compactenum}
\item [(M1)] $\Psi (x) =\sgn (x) \quad \text{for all} \;\; |x|\geq
  \delta ;$ 
\item [(M2)] $\psi (x):=\Psi ' (x) >0 \quad \text{for all} \;\; x\in
  (\delta ,\delta); $  
\item [(M3)] $\psi ' (x)>0  \quad \text{for all} \;\; x\in (-\delta
  ,0)$; and 
\item [(M4)] $\psi(x)$ has a unique inflection point $x=\iota$ in
  $(-\delta,0)$, such that
  \begin{compactenum}
    \item [(M4a)] $\psi'(x)$ is strictly increasing in
      $[-\delta,\iota]$, and 
    \item [(M4b)] $\psi'(x)$ is strictly decreasing in $[\iota ,0]$.
  \end{compactenum}

\end{compactenum}
Thus, the function $\frac{1}{2} \psi$ is an even probability density
with support $[-\delta,\delta]$.  We shall refer to $\Psi$ as the
\emph{payoff function}, because it determines the quantity by which
values $u_i$ are multiplied to obtain the allocative component of each
individual's utility.\footnote{The assumptions on the payoff function
  $\Psi$ are primarily for mathematical convenience. However, there
  are some circumstances where a smoothing of the payoff for vote
  totals near $0$ might be natural: for instance, (i) in some close
  elections, it might be necessary for the winning side to form a
  coalition with some of the losers to form a functioning majority; or
  (ii) for vote totals near $0$, a recount might be necessary, leading
  to the possibility that the winning side might be overturned.}
Conditional on the values $\{v_i\}$, individual $i$ earns expected
utility
\begin{equation}\label{utility}
	 \Psi(V)u_{i}-v_i^2.
\end{equation}
Thus, in a \emph{type-symmetric, pure-strategy Bayes-Nash
  equilibrium}\footnote{ See
  section~\ref{ssec:assumptions-terminology} for the definition and a
  proof that Bayes-Nash equilibria use non-randomized
  strategies. Roughly, a type-symmetric equilibrium is a function
  $v(u)$ such that, if all players use the rule $u\mapsto v (u)$ for
  buying votes then no player could improve her expected utility by
  defecting from the strategy.}, a voter with value $u$ will maximize
\begin{equation}\label{eq:wealth-maximization}
	 E\left[u\Psi\left(S_{n}+v\right)\right]-v^{2},
\end{equation}
where $n=N-1$ and $S_{n}:=\sum_{ i= 1}^{n}v_{i}$ is the {\em one-out
  vote total}, the sum of all votes cast by all but a single
individual.  For brevity, we shall refer to \emph{type-symmetric,
  pure-strategy Bayes-Nash equilibria} as \emph{Nash equilibria}.

We define the {\em expected inefficiency} as 
\[
EI\equiv \frac{1}{2} - \frac{ E\left[ U \Psi(V)
\right]}{ 2 E\left[ \left| U\right| \right]} \in[0,1],
\]
where $U\equiv
\sum_i u_i$.  This measure is the unique negative monotone linear
functional of realized aggregate utility  $ U \Psi(V) $ that is
normalized to lie in the unit interval.

\subsection{Existence of Equilibria}\label{existence}

\begin{proposition}\label{proposition:existence}
\label{Lem:existence}For any $N>1$ a monotone increasing,
pure-strategy Nash Equilibrium $v$ exists.
\end{proposition}

This result follows directly from \cite{renybayesian}, Theorem
4.5. All of \cite{renybayesian}'s conditions can easily be
checked, so we highlight only the less obvious ones.  Continuity of
payoffs as functions of the  actions $v_{i}$ follows from the
continuity and boundedness of
$\Psi$.  Type-conditional utility is only bounded from above, not
below, but boundedness from below can easily be restored by simply
deleting for each value type $u$ votes of magnitude greater $\sqrt{2
\left| u \right|}$.  The existence of a monotone best-response follows
from the obvious super-modularity of payoffs in value and votes. 

Although Nash equilibria always exist, they need not be
unique. Indeed, we will show that in some circumstances
(cf. Theorem~\ref{munot0characterization}) Nash equilibria have points $u_{*}$
of discontinuity; at any such point, there are at least two distinct
pure-strategy Nash equilibria, one with $v (u_{*})=v (u_{*}+)$, the
other with $v (u_{*})=v (u_{*}-)$. We conjecture, however, that at
least when $N$ is large, non-uniqueness of Nash equilibria can only
occur for this trivial reason: in particular, we conjecture that if
$v_{1}$ and $v_{2}$ are distinct Nash equilibria then $v_{1} (u)=v_{2}
(u)$ for all but at most one value $u$.

\subsection{Rationale for QV}\label{intuition} Formally differentiating expression
\eqref{eq:wealth-maximization} with respect to $v$ (see
section~\ref{ssec:nec-conditions} for a formal proof) yields the following
first-order condition for maximization:
\begin{equation}\label{foc}
 u  E\left[\psi\left(S_{n}+v\right)\right] = 2v \implies
v(u)=\underbrace{\frac{E\left[\psi\left(S_{n}+v(u)\right)\right] }{2}}_{\textrm{marginal
pivotality}} u.  
\end{equation} 
The marginal benefit of an additional
unit of vote is thus twice the individual's value multiplied  by the influence
this extra vote has on the chance the alternative is adopted, the
vote's {\em marginal pivotality}.  The marginal cost of a vote is
twice the number of votes already purchased.

When the number $N$ of voters is large, most would reason that their votes
$v (u)$ will have a negligible effect on the vote total $S_{n}+v
(u)$.  Taking this logic to an extreme, if voters acted as if marginal
pivotality $p$ were constant across the population, then an individual
with value $u$ would buy $v(u)=pu$ votes. This voting strategy would
imply $V=p\sum_i u_i$; that is, the vote total would be exactly
proportional to the sum of the values, and consequently the expected
inefficiency would be $0$.  Clearly, this argument holds only for a
quadratic cost function, because only quadratic functions have linear
derivatives. 

Our main results will show, however, that the marginal pivotality is
not constant; in fact, when the mean $\mu$ of the sampling
distribution $F$ is non-zero the marginal pivotality can have large
jump discontinuities in the tail of the distribution. Thus, voters do
not buy votes strictly in proportion to their values, and so in
general the vote total will not always be a scalar multiple of the aggregate
value $\sum_{i}u_{i}$. Nevertheless, as our results will show,
quadratic voting is asymptotically efficient, in the sense that the
expected inefficiency converges to $0$ as $N \rightarrow\infty$.

Although it is perhaps obvious, we emphasize that one-person-one-vote
is in many cases \emph{not} efficient. Such inefficiency will occur,
for instance, if the distribution $F$ has positive mean $\mu$ but attaches
probability $q< 1/2$  to the interval $[0,\overline{u} ]$, because by
the law of large numbers, when $N $ is large,
\begin{displaymath}
  \frac{1}{N}\sum_{i=1}^{N}U_{i} \approx \mu >0
\end{displaymath}
but
\begin{align*}
  \frac{1}{N}\sum_{i=1}^{N} \mathbf{1}_{\xset {U_{i}\geq 0}}&\approx q
  \quad \textrm {and}\\
  \frac{1}{N}\sum_{i=1}^{N} \mathbf{1}_{\xset {U_{i}< 0}}&\approx 1-q,
\end{align*}
so under one-person-one-vote the vote total would, with high
probability, be near $(-1+2q )N <0$.

\subsection{Main Results}\label{main}

Our main results concern the structure of equilibria in the game
described in the previous section when the number $N$ of agents is
large, and the implications for the efficiency of QV.

\subsubsection{Characterization of equilibrium in the zero mean
case}\label{balancedsample}

The structure of a Nash equilibrium differs
radically depending on whether $\mu=0$ or $\mu\neq0$.  The case   $\mu
=0$  is of particular interest because in some
elections -- for instance, when two candidates are vying for an elected
office -- the alternatives may be tailored so that an approximate
population balance is achieved \cite{ledyardvote}.

\begin{theorem}\label{mu0characterization}
For any sampling distribution $F$ with mean $\mu =0$ that
satisfies the hypotheses above, there exist
constants $\epsilon_{N} \rightarrow 0$  such that in any Nash
equilibrium, $v (u)$ is $C^{\infty}$ and strictly increasing on $[\underline u 
,\overline u ]$ and satisfies the following approximate
proportionality rule:
\begin{equation}\label{eq:proportionality-rule}
	\bigg| \frac{v (u)}{p_N u}-1 \bigg| \leq
	\epsilon_{N}
	\quad \text{where} \quad
	p_{N}=\frac{1}{2^{\frac{3}{4}}\sqrt{\sigma}\sqrt[4]{\pi(N-1)}}.
\end{equation}
Furthermore, there exist constants $\alpha_{N},\beta_{N}\rightarrow 0$ 
such that in any equilibrium the vote total $V=V_{N}$ and expected
inefficiency satisfy
\begin{gather}\label{eq:e-vote-total}
	| E[V]|\leq \alpha_{N} \sqrt{\var (V)} \quad \text{and}\\
\label{eq:eff}
	\quad EI<\beta_{N}.		
\end{gather}
\end{theorem}

The proof will be given in section \ref{sec:balanced}.

Thus, in any equilibrium, agents buy votes \emph{approximately} in
proportion to their values $u_{i}$.  Given this fact, it is not
difficult to understand why the number of votes a typical voter buys
should be of order $N^{-1/4}$. If the vote function $v (u)$ in a Nash
equilibrium follows a proportionality rule $v (u)\approx \beta u$, the
constant $\beta$ must be the consensus marginal pivotality. On the
other hand, by the local limit theorem of probability (see
\cite{feller}, ch.~XVI), if $\beta =CN^{-\alpha}$ for some constants
$C\not =0$ and $\alpha \in \zz{R}$, the chance that
$V\in [-\delta ,\delta]$ would be of order $N^{\alpha -\frac{1}{2}}$,
and so $\alpha$ must be $1/4$.

Although the relation \eqref{eq:proportionality-rule} asserts the
ratio $v (u)/u$ is approximately constant, it is not \emph{exactly}
constant: in fact, $v (u)$ is a genuinely nonlinear function of
$u$. Thus, even though $ E[U]=0$, it need not be the case that
$E[V]=0$.  To establish the asymptotic efficiency assertion
\eqref{eq:eff}, we must establish assertion \eqref{eq:e-vote-total},
namely, that the non-linearities vanish rapidly enough that the bias
created by non-linearity is smaller than the sampling variation in
$u$.  This will require a rather subtle application of the {\em
  Edgeworth expansion} (cf. \cite{feller}, ch.~XVI) of the
distribution of $S_{n}$.  If it were the case that $ E[V]=0$, and if
the distribution of $S_{n}$ were exactly normal, a standard Taylor
expansion and the $N^{-1/4}$ decay of $v (u)/u$ could be used directly
to show that non-linearities vanish with $N^{-1}$ even relative to the
leading term of $v (u)/u$.  A detailed analysis of this argument leads
us to conjecture that, under the hypotheses of
Theorem~\ref{mu0characterization}, the inefficiency of QV decays like
$\mu_{3}^{2}/ (16\sigma^{6}N)$.

\subsubsection{Characterization of equilibrium in the non-zero mean
case}\label{unbalancedsample}

When $\mu$ is not zero, the nature of equilibria can be quite
different: in particular, if the payoff function $\Psi$ is
sufficiently sharp (i.e., the support of its derivative $\psi$ is
sufficiently small) then for sll large $N$, every Nash equilibrium has
a large discontinuity in the extreme tail of the sampling
distribution. Nevertheless, in all cases the quadratic voting
mechanism is asymptotically efficient, as the following theorem shows.

\begin{theorem}\label{theorem:efficiency}
  Assume that the sampling distribution $F$ has mean $\mu>0$ and that
  $F$ and $\Psi$ satisfy the hypotheses laid out in
  section~\ref{ssec:model} above. Then there exist
  constants $\beta_{N}\rightarrow 0$ such that in any Nash equilibrium
  $v (u)$,
\begin{equation}\label{eq:asym-eff}
	EI < \beta_{N}.
\end{equation}
Furthermore, there exist constants $\xi \geq \delta$ and $\beta>0$
depending on the sampling distribution $F$ and the payoff function
$\Psi$ but not on $N$ such that in any equilibrium $v (u)$, for any
$\epsilon >0$,
\begin{gather}\label{eq:approx-proportionality}
	\sup_{\underline u +\beta N^{-3/2}\leq u\leq \overline u}\bigg|v
        (u)-
        \left(
          \frac{\xi}{\mu N}
        \right)u\bigg|
	<\epsilon_{N}/N \quad \text{and hence}\\
	\label{eq:vote-total-concentration}
	P\{|V_{N}-\xi |>\epsilon \}\leq \epsilon'_{N},
\end{gather}
where $\epsilon_{N},\epsilon'_{N}\rightarrow 0$ are constants that depend
only on the sample size $N$, and not on the particular equilibrium.
\end{theorem}

This theorem allows for two cases.  In the first, where $\xi=\delta
$, the vote total is near $\delta $ with high probability for large
$N$.  This case occurs for large $\delta$ and thus relatively smooth
payoff functions.  In the second, $\xi>\delta$, so that with high
probability the vote total is outside $\left[-\delta, \delta\right]$
for large $N$. This case arises when $\delta >0$ is small (see
Proposition~\ref{proposition:min-problem} below). In both
cases, the approximate proportionality rule
\eqref{eq:approx-proportionality} holds except possibly in the extreme
lower tail of the value distribution $F$.

 To see how the dichotomy arises, suppose that for some $\xi \geq
\delta$ there were a value $w\in (-\delta ,\delta)$ such that
\begin{equation}\label{eq:min-probability-a}
	(1-\Psi (w))\left|\underline u\right|> (\xi  -w)^{2} ;
\end{equation}
then an agent with value $u$ near the lower extreme $\underline
u$, knowing that with high probability the one-out vote total
$S_{n}$ is near $\xi $, would find it
worthwhile to buy $-\xi +w$ votes and thus single-handedly move the
vote total to $w$.  Consequently, there can be no equilibrium in which
$S_{n}$ concentrates strictly below $\xi$ if such a $w$ exists, as
this would lead a large number of individuals to act as extremists,
contradicting the concentration of the vote total. Therefore, in any
equilibrium the voters with positive values $u_{j}$ must buy enough
votes to guarantee that the vote total concentrates at or above
$\xi$. The minimal value $\xi\geq \delta$ below which there is no  advantage
to ``extremist'' behavior in the extreme lower tail  thus
determines the equilibrium behavior \eqref{eq:approx-proportionality}.
This will be at $\xi =\delta$ \emph{unless} there is a solution to
the following problem.

\bigskip \noindent
\textbf{Optimization Problem.} Determine $\xi >\delta$
and a matching real number $w\in [-\delta ,\delta]$  such that
\begin{align}\label{eq:xi-w}
	 (1-\Psi (w))\left|\underline u\right|&= (\xi  -w)^{2}
	 \quad \text{and}  \\
\notag 	  (1-\Psi (w'))\left|\underline u\right|&<   (\xi
	   -w')^{2} \quad \text{for all}\; w'\in [-\delta ,\delta ]\setminus \{w \}
\end{align}

\begin{proposition}\label{proposition:min-problem}
If $\delta <1/\sqrt{2}$ then there exists a unique pair $\xi >\delta$
and  $w\in (-\delta ,\delta )$ that satisfy the
Optimization Problem \eqref{eq:xi-w}.
\end{proposition}

  The proof will be given in Appendix~\ref{sec:min-problem}. When the
Optimization Problem has a solution, Nash equilibria take a rather
interesting form in which extremists must appear, but with vanishing
probability, as the following theorem shows.

\begin{theorem}\label{munot0characterization}
  Assume that the sampling distribution $F$ has mean $\mu>0$ and that
  $F$ and $\Psi$ satisfy the hypotheses above.  Assume further that
  the Optimization Problem \eqref{eq:xi-w} has a unique solution
  $(\xi ,w)$ such that $\xi >\delta$. Then there exists a constant
  $\zeta >0$ depending on $F$ such that for any $\epsilon >0$ and any
  Nash equilibrium $v(u)$, when $N$ is sufficiently large, there
  exists $u_{*}\in [\underline{u} . \overline{u} ]$ such that
\begin{compactenum}
\item [(i)] $v (u)$ has a jump  discontinuity at $u=u_{*}$;
\item [(ii)] $v(u)$ is continuous and continuously differentiable for
  all $u\in (u_{*},\overline{u} ]$;
\item [(iii)]  the approximate proportionality rule
\eqref{eq:approx-proportionality} holds for all $u\in (u_{*},\overline
u]$; 
\item [(iv)] $|v (u) +\xi -w| <\epsilon$  for $u\in
\left[\underline u,u_{*}\right)$; and
\item [(v)] $|u_{*} -(\underline{u} +\zeta N^{-2})|<\epsilon N^{-2}$.
\end{compactenum}
\end{theorem}

Theorems \ref{theorem:efficiency} and \ref{munot0characterization}
will be proved in
section~\ref{sec:unbalanced}. Theorem~\ref{munot0characterization}
asserts that any Nash equilibrium has a single large discontinuity
near $\underline{u} +\zeta N^{-2}$; it does not preclude the possibity
of other discontinuities in the interval $[\underline{u}, u_{*} )$,
but (iv) implies that if these occur, the jumps must be small. 

Theorem~\ref{theorem:efficiency} implies that an agent with value $u$
will buy approximately $\xi \mu ^{-1}u/N$ votes {unless} $u$ is in
the extreme lower tail of $F$. Theorem~\ref{munot0characterization}
implies that when a solution to the Optimization Problem exists, such
exceptional agents occur only with probability
$\approx \zeta N^{-1}f (\underline u)$; consequently, by the law of large
numbers, with probability
$\approx 1-\zeta N^{-1}f (\underline u)$ the vote total will be very
near $\xi $.  If, on the other hand, the sample contains an agent
with value less than $u_{*}$ then this agent will buy approximately
$w-\xi \approx -\sqrt{|\underline u |}$ votes, enough to move the
overall vote total close to $w$. Agents of the first type will be
called \emph{moderates}, and agents of the second kind {\em extreme
  contrarians} or {\em extremists} for short. Because the tail region
in which extremists reside has $F-$probability on the order $N^{-2}$,
the sample of agents will contain an extremist with probability only
on the order $N^{-1}$, and will contain two or more extremists with
probability on the order $N^{-2}$. Given that the sample contains no
extremists, the conditional probability that $|V-\xi |>\epsilon $
is $O (e^{-\varrho n})$ for some $\varrho >0$, by standard large
deviations estimates, and so the event that $V<0$ essentially
coincides with the event that the sample contains an extremist.

Why does equilibrium take the somewhat counter-intuitive form
described in Theorem \ref{munot0characterization}? Following is a
brief heuristic explanation. For an agent 
with value $u$ in the ``bulk'' of the sampling distribution $F$,
there is very little information about the vote total $V$ in the
agent's value, and so for most such agents the marginal
pivotality $ \frac{1}{2} E\psi \left(S_{n}+v(u)\right)$
will be approximately $\frac{1}{2}  E\psi (V)$. Consequently, $v(u)$ will be
approximately linear in $u$ except possibly in the extreme tails of
the distribution, and so by  the law of large numbers,
the vote total will, with high probability, be near $\frac{1}{2} N \mu E
\psi (V)$.

Because $\mu >0$, agents with negative values will, with high
probability, be on the losing side of the election. However, if
$ \frac{1}{2} N\mu E \psi (V)$ were small, then an agent with even
moderately negative value could increase her expected utility by
buying  a large number of  (negative) votes;
since many voters with negative values would find it beneficial to
adopt such a strategy, the vote total would, with high probability, be
negative, in contradiction to the fact that it must be concentrated
near $\frac{1}{2} N \mu E \psi (V)$.Therefore, $N E\psi (V)\mu $ must
remain bounded away from $0$.

On the other hand, if $\frac{1}{2} N\mu  E \psi (V)$ were too large,
then  no individual could profitably act as an
extremist, so except with exponentially small probability the  vote
total $V$ would be
bounded away from $[-\delta ,\delta]$. But this would force  $\mathbb E \psi
(V)$ to be exponentially small, which is impossible. Thus,
the aggregate number of votes must concentrate near a constant value,
and so most voters must buy on the order of $1/N$ votes. For this
scenario to occur, $ E \psi (V)$ must decay as
$\frac{1}{N}$. But the primary contribution to this expectation
must come from the event in which an extremist exists, and so the
probability of this event must decay as $\frac{1}{N}$.

\subsection{Proofs of the Main Results: Two Key Ideas}

\emph{A priori}, we know almost nothing about a Nash equilibrium $v(u)$,
apart from monotonicity and boundedness. The heuristic arguments
advanced in section~\ref{intuition} above suggest that $v(u)$ should
be at least approximately linear in $u$, but on close inspection these
arguments fail to hold water: in particular, a large discontinuity in
$v(u)$ might also produce a large discontinuity in the marginal
pivotality (see equation~\eqref{foc}), and so the equality \eqref{foc}
might persist across such a discontinuity. The major part of the
analysis below will focus on the problems of determining where -- or if
-- such discontinuities can occur, and  how much variation in the
marginal pivotality there can be in  regions free of
discontinuities. Two key ideas, both based in the mathematical laws
governing random sampling, will prove central to this analysis.

\bigskip
\textbf{Weak Consensus:} The first 
idea concerns sampling (with replacement) from a multinomial
population with $K$ categories. Let $p_{i}$ be the population
frequency  of category $i$, and let $N_{i}$ be the number of
individuals in a random sample of size $N \gg K$ that fall in category
$i$. Then for any two count vectors  $\mathbf{n}=(n_{i})_{i \leq K}$
and $\mathbf{n}'=(n'_{i})_{i \leq K}$ that differ by at most one (in
absolute value) in each slot and satisfy
\begin{displaymath}
  |n_{i}-Np_{i}|\leq N \epsilon _{N}\quad \textrm{for each} \; i=1,2,
  \cdots  ,K
\end{displaymath}
where $\epsilon_{N}\rightarrow 0$ as $N \rightarrow \infty$, we have
\begin{displaymath}
  \frac{P \xset {N_{i}=n_{i} \; \textrm{for each} \; i}}{P \xset
    {N_{i}=n'_{i} \; \textrm{for each} \; i}} \approx 1.
\end{displaymath}
Moeover,  this approximation holds \emph{uniformly} for all
probability vectors $\mathbf{p}=(p_{i})_{i \leq K}$ with
minimum entry bounded away from $0$. It should be noted that the
approximation breaks down when one of the categories has
probability near $0$. The upshot is that for two agents with values
$u_{1},u_{2}$ \emph{not} in the tail of the value distribution $F$,
the distributions (conditional on their values $u=u_{1},u_{2}$,
respectively) of the vote total
$S_{n}+v(u)$ will be nearly the same, because by the multinomial
sampling principle moving one agent from
the ``category'' $[u_{1}-\epsilon, u_{1}+\epsilon]$ to the
``category'' $[u_{2}-\epsilon, u_{2}+\epsilon]$ will not change
sampling probabilities appreciably. However, for an agent with value
$u_{3}$ in the extreme tail of the  distribution this approximate
consensus breaks down, because the ``category'' $[u_{3},\bar{u}]$
has sampling probability near zero. See the proof of
Lemma~\ref{lemma:consensusA}  for a precise formulation of this
assertion. 

\bigskip
\textbf{Anti-Concentration.} The second idea concerns concentration
properties of the distribution of a sum of independent, identically
distributed random variables, and originates in the (local) central
limit theorem. A precise formulation is given in
Proposition~\ref{proposition:anti-concentration} in
section~\ref{ssec:size} below; following is a rough discussion. If
$X_{1},X_{2}, \cdots  $ are independent, identically distributed
random variables with variance $\sigma^{2}<\infty$ and centered third
moment bounded by $C \sigma^{3/2}$, then for any interval $J \subset
\mathbb{R}$ of length $2\delta$,
\begin{displaymath}
  P \xset {\sum_{i=1}^{n}X_{i}\in J} \leq \frac{2C'\delta}{\sigma \sqrt{n}} ,
\end{displaymath}
where $C'$ is a finite constant depending only on $C$. Now the function $\psi$ has,
by assumption, support contained in the interval $[-\delta,\delta]$;
consequently, for any Nash equilibrium $v(u)$ such that
$|v(u_{1})|\geq \epsilon$ for \emph{some} $u_{1}$, we must have
\begin{displaymath}
  \epsilon \leq \frac{1}{2} E\psi(S_{n}+v(u_{1})) \leq
  \frac{\xnorm{\psi}_{\infty}}{2}P \xset {S_{n}\in [-v(u_{1})-\delta,
    -v(u_{1})+\delta] } \leq \frac{C'\delta}{\sqrt{n \textrm {var}(v(U_{1}))}}.
\end{displaymath}
Thus, the size of $\textrm {var}(v(U_{1})$ is severely constrained,
and so a relatively large  $|v(u_{1})|$ can occur only if \emph{most}
voters buy fewer than $O(1/\sqrt{n})$ votes. Variations of this
argument will be used in section~\ref{ssec:size} to limit the
locations of discontinuities of Nash equilibria $v(u)$ to the extreme
tails of the value distribution $F$ (see, in particular,
Proposition~\ref{proposition:max-cont-interval}).

We remark that this argument depends crucially on the as-yet unproven
fact that the third centered moment of $v(U_{1})$ is bounded by $C
\textrm {var}(v(U_{1}))^{3/2}$. This is where the weak consensus
bounds discussed above will prove critical, as they will enable us to
obtain uniform bounds on the third moments for any Nash equilibrium $v(u)$.

\subsection{Plan of the paper}\label{ssec:plan} The remainder of the
paper will be devoted to the proofs of Theorems 1--3 and
Proposition~\ref{proposition:min-problem}.  Because essentially
nothing (other than monotonicity) is known \emph{a priori} about the
nature of Nash equilibria, information must be teased out in steps,
each relying on the previous steps. We begin in
section~\ref{sec:nec-conditions} by collecting some relatively easy
consequences of monotonicity and the first-order necessary
condition~\eqref{eq:weierstrass}, including a useful necessary
condition (section~\ref{ssec:discontinuities}) for discontinuities of
Nash equilibria, which will ultimately be used to prove that these can
occur only in the extrreme tails of the distribution $F$.  In
section~\ref{ssec:consensus}, a weak form of the approximate
proportionality rule will be proved for agents in the bulk of the
distribution $F$. Using this weak approximate proportionality rule, we
will, in section~\ref{ssec:size}, use an anti-concentration inequality
for sums of i.i.d. random variables to derive bounds for Nash
equilibria. We will then be able to deduce, in
section~\ref{sec:app-prop}, that approximate proportionality holds
except in the extreme tails of $F$. The proofs of Theorems 2--3 will
be given in section~\ref{sec:unbalanced}, and the proof of Theorem~1
in section~\ref{sec:balanced}. 

\subsection{Notation}\label{ssec:assumptions-terminology}
The symbols $\Psi,\psi $, $\delta $, $F$, $f$, $\mu$, $\sigma^{2}$,
$\xi$, $\zeta $, $\iota$, $\underline u,\overline u$ will be reserved
for the functions and constants specified in
Subsection~\ref{ssec:model} above, and the letters $N,n$ will be used
only for the sample size and sample size minus one. The symbols
$\alpha ,\beta ,\gamma ,\epsilon ,\varrho $ and $C, C', \cdots $ will
be used for generic constants whose values might change from one lemma
to the next. Because many of the arguments will involve the values of
the equilibrium vote function $v$ at points near one of the endpoints
$\underline u , \overline u$, we will use the following shorthand
notation, for any $0<\epsilon <1$:
\[
	\overline u_{\epsilon}=\overline u -\epsilon \quad \text{and}
	\quad 
	\underline u_{\epsilon}=\underline u +\epsilon .
\]

\section{Nash Equilibria: Basic Properties}\label{sec:nec-conditions}

\subsection{Monotonicity of Nash Equilibria}\label{ssec:monotonicity}

\begin{proposition}
  \label{proposition:strict-monotonicity} Any pure-strategy Nash
  equilibrium $v(u )$ is \emph{strictly} increasing in $u$.
\end{proposition}

Monotonicity of Nash equilibria has already been established
(cf. Proposition~\ref{proposition:existence}); this follows from
general results in game theory. Strict monotonicity requires an
additional argument. Because this is relatively standard, we relegate
the proof to Appendix~\ref{sec:monotonicity}.

\subsection{Necessary Condition for a Nash
  Equilibrium}\label{ssec:nec-conditions}

\begin{proposition}\label{proposition:weierstrass}
If $v(u)$ is a  Nash equilibrium then at every $u \in
[\underline{u},\overline{u} ]$  the function $v$ 
satisfies  the functional equation
\begin{equation}\label{eq:weierstrass}
	E\psi (S_{n}+v (u))u=2 v (u).
\end{equation}
\end{proposition}

The necessary condition \eqref{eq:weierstrass} will be of central
importance in the analysis to follow.  The proof, which is both easy
and completely standard, is given in
Appendix~\ref{sec:weierstrass}. Observe, though, that the proposition
depends crucially on the differentiability of the payoff function $\Psi$; for
functions with discontinuities, such as
$\Psi =\mathbf{1}_{[0,\infty)}-\mathbf{1}_{(-\infty,0)}$, the argument
breaks.

\subsection{Discontinuities of Nash Equilibria}\label{ssec:discontinuities}

Because any Nash equilibrium $v (u)$ is a monotone functions of $u$,
it can have at most countably many discontinuities, all of which are
jumps. Clearly, any Nash equilibrium $v(u)$ is continuous at $u=0$,
because for very small $|u|$ an agent with value $u$ would never pay
more than $2|u|$ for votes, since this is the maximal change in the
agent's utility that could result. The following proposition asserts
that there is a lower bound on the magnitude $|v(u)|$ of a Nash
equilibrium near any point of discontinuity.

\begin{proposition}
  \label{proposition:discontinuities-necessary}
There exists $\Delta >0$  such that for all sufficiently large
$n$, at any point $u_{*}$ of discontinuity of a Nash equilibrium,
\begin{equation}\label{eq:discontinuities-size}
	\limsup_{u \rightarrow u_{*}}|v (u)| \geq \Delta .
\end{equation}
\end{proposition}

By the monotonicity of $v$, the limsup must be  the limit of
$v(u)$ as either $u \downarrow u_{*}$ or as $u \uparrow u_{*}$,
depending on whether $u_{*}$ is negative or positive. Thus, at any point
$u_{*}$ of discontinuity, $|v(u)|$ must jump to at least $\Delta$ as
$u$ passes through the value $u_{*}$.  \emph{A priori}, we have no
information about the size of a Nash equilibrium; however, in
sections~\ref{ssec:consensus} and \ref{ssec:size} we will show that large
values of $|v(u)|$ can only occur in the extreme tails of the sampling
distribution $F$ (in particular, within distance $O(n^{-3/2})$ of one
of the endpoints $\underline{u} , \overline{u} $).

The proof of Proposition~\ref{proposition:discontinuities-necessary}
will require two auxiliary lemmas, both of which will be of use later
in the paper.

\begin{lemma}\label{lemma:discontinuities-necessary}
Let $v (u)$ be a Nash equilibrium. If $v$ is discontinuous
at $u\in \left(\underline u,\overline u\right)$ then
\begin{equation}\label{eq:discontinuities-necessary}
	E\psi ' (\tilde{v}+S_{n})u=2
\end{equation}
for some $\tilde{v}\in [v_{-},v_{+}]$, where $v_{-}$ and $v_{+}$ are
the left and right limits of $v (u')$ as $u' \rightarrow u$.
\end{lemma}

\begin{proof}
  The necessary condition \eqref{eq:weierstrass} holds at all $u'$ in
  a neighborhood of $u$, so by monotonicity of $v$ and continuity of
  $\psi$, Equation \eqref{eq:weierstrass} must hold when $v (u)$ is
  replaced by either of $v_{\pm}$, that is,
  \begin{align*}
    	2v_{+}&=E\psi (v_{+}+S_{n})u \quad \text{and}\\
	2v_{-}&=E\psi (v_{-}+S_{n})u.
  \end{align*}
  Subtracting the second equation from the first and using the
  differentiability of $\psi$ we obtain
  \begin{displaymath}
    2v_{+}-2v_{-}=uE\int_{v_{-}}^{v_{+}} \psi ' (t+S_{n}) \,dt
    =u\int_{v_{-}}^{v_{+}}E \psi ' (t+S_{n}) \,dt . 
  \end{displaymath}
  The lemma now follows from the mean value theorem of calculus.
\end{proof}

\begin{lemma}\label{lemma:psi-psi-prime}
For any $\alpha >0$  there exists $\beta =\beta ( {\alpha})>0$ 
such that for any Nash equilibrium $v (u)$, any $\tilde{v}\in \zz{R}$, any
$u\in \left[\underline u,\overline u\right]$, and all  $n$,
\begin{align}\label{eq:psi-psi-prime}
	E|\psi ' (\tilde{v}+S_{n})u |\geq \alpha  \quad
	&\Longrightarrow \quad
	E\psi  (\tilde{v}+S_{n})|u| \geq \beta \;\;\text{and}\\
\notag 	E|\psi '' (\tilde{v}+S_{n})u |\geq \alpha  \quad
	&\Longrightarrow \quad
	E\psi  (\tilde{v}+S_{n})|u| \geq \beta.
\end{align}
\end{lemma}

\begin{proof}
 Recall that $\psi /2$ is a $C^{\infty}$ probability density with
support $[-\delta ,\delta]$, and  that $\psi$ is \emph{strictly}
positive in the open interval $(-\delta ,\delta)$. Consequently, on
any interval $J\subset (-\delta ,\delta)$ where $|\psi '|$ (or $|\psi
''|$) is bounded below by a positive number, so is $\psi$.

Fix $\epsilon >0$ so small that
$\epsilon \max\left(|\underline u|,\overline u\right) <\alpha /2$. In
order that $E|\psi ' (\tilde{v}+S_{n})u| \geq \alpha $, it must be the
case that the event $\{|\psi ' (\tilde{v}+S_{n})|\geq \epsilon \}$
contributes at least $\alpha /2$ to the expectation; hence,
\begin{displaymath}
  P\{|\psi ' (\tilde{v}+S_{n})|\geq \epsilon \}\geq \frac{\alpha
  }{2\xnorm{\psi '}_{\infty}\max(|\underline{u}|, \overline{u}}. 
\end{displaymath}
But on this event the random variable $\psi (\tilde{v}+S_{n})$ is
bounded below by a positive number $\eta =\eta_{\epsilon}$, so it
follows that
\begin{displaymath}
  E\psi  (\tilde{v}+S_{n})|u| \geq \frac{\eta \alpha
        }{2\xnorm{\psi '}_{\infty}}.
\end{displaymath}
A similar argument proves the corresponding result for $\psi ''$.
\end{proof}

\begin{proof}[Proof of
  Proposition~\ref{proposition:discontinuities-necessary}]
  Without loss of generality, we can assume that the point $u_{*}$ of
  discontinuity is positive, because if necessary we can exchange the
  roles of positive and negative voters. By
  Lemma~\ref{lemma:psi-psi-prime},
  equation~\eqref{eq:discontinuities-necessary} implies that
     \begin{displaymath}
       E\psi (\tilde{v}+S_{n})u_{*}\geq \beta=\beta(2)>0  
     \end{displaymath}
     for some $\tilde{v}\in[v_{-},v_{+}]$, where $v_{-},v_{+}$ are the
     left and right limits of $v(u)$ at $u_{*}$ and $\beta(\alpha)>0$
     is the constant in \eqref{eq:psi-psi-prime}.
     Consequently, either
     \begin{align*}
       E\psi(v_{-}+S_{n}) u_{*}& \geq \beta/2 \quad \textrm {or}\\
       (E\psi (\tilde{v}+S_{n}) -E\psi(v_{-}+S_{n}) )u_{*}& \geq \beta/2.
     \end{align*}
     In the former case, we must have $\lim_{u \rightarrow
       u_{*}-}2v(u)\geq \beta/2$, by the necessary condition
     \eqref{eq:weierstrass}. 
     In the latter case,
     \begin{align*}
       \beta/2 &\leq (E\psi (\tilde{v}+S_{n})-E\psi(S_{n}) )u_{*}\\
                              &=u_{*}E\int_{v_{-}}^{\tilde{v}} \psi'(v+S_{n}) \,dv \\
       &\leq \overline{u} \xnorm{\psi'}_{\infty} (\tilde{v}-v_{-}),
     \end{align*}
     which implies that $v_{+}\geq \tilde{v}\geq \tilde{v}-v_{-} \geq
     \beta/(4\xnorm{|\psi'|}_{\infty}\overline{u} )$.
\end{proof}

\subsection{Smoothness}\label{ssec:smoothness}

Because Nash equilibria are monotone, by
Lemma~\ref{proposition:existence}, they are necessarily differentiable
almost everywhere. The following proposition gives more precise
information about points of non-differentiability. Together with the
results of sections~\ref{ssec:consensus} and \ref{ssec:size} below,
this proposition will imply that Nash equilibria must be smooth except
in the extreme tails of $F$.

\begin{proposition}\label{proposition:smoothness}
  A Nash equilibrium $v(u)$ is $C^{\infty}$ at every $u\in
  (\underline{u} ,\overline{u} )$ except those $u$ at which $v(u)$ is
  discontinuous or at which
  \begin{equation}
    \label{eq:bad-u}
    E\psi'(v(u)+S_{n})u=2.
  \end{equation}
  At any point $u\in(\underline{u} ,\overline{u})$ where $v(u)$ is
  differentiable, the first and second derivatives are determined  by
  the equations
  \begin{align}
    \label{eq:first-derivative}
    2v'(u)&=E\psi(v(u)+S_{n})+E\psi'(v(u)+S_{n})v'(u)u,\\
    \label{eq:second-derivative}
    2v''(u)&=E\psi'(v(u)+S_{n})(2v'(u)+v''(u)u)+E\psi''(v(u)+S_{n})v'(u)^{2}u.    
  \end{align}
  
\end{proposition}

\begin{proof}
  The necessary condition \eqref{eq:weierstrass} can be written as
  $H(v(u),u)=0$ where $H$ is the function of two variables
  defined by
  \begin{displaymath}
    H(v,u)=E\psi (v+S_{n})u - 2v.
  \end{displaymath}
  Since $\psi$ is by hypothesis $C^{\infty}$ and has compact support,
  $H$ is also $C^{\infty}$, with first partial derivative
  \begin{displaymath}
    \frac{\partial H}{\partial v}=E\psi'(v+S_{n})u -2
  \end{displaymath}
  Hence, the implicit function theorem of calculus implies that every point
  $(u,v(u))$ where equation \eqref{eq:bad-u} does not hold has an open
  neighborhood in the plane in which the solution set of the equation
  $H(v,u)=0$ defines a $C^{\infty}$ function $v(u)$. The necessary
  condition  \eqref{eq:weierstrass} implies that this function must
  coincide locally with the Nash equilibrium unless the equilibrium has a
  discontinuity at $u$.

  Since the left side of equation \eqref{eq:bad-u} is continuous in
  $u$, any point $u$ where \eqref{eq:bad-u} fails has an open
  neighborhood in which \eqref{eq:bad-u} fails, and hence in which
  $v(u)$ is differentiable.  The equations \eqref{eq:first-derivative}
  and \eqref{eq:second-derivative} for the first two derivatives
  follow from \eqref{eq:weierstrass} by the chain rule
  (differentiation under the expectation is justified because $\psi$
  is $C^{\infty}$ with bounded support).
\end{proof}

The differential equation \eqref{eq:first-derivative} can be rewritten
as
\begin{equation}
    \label{eq:1st-derivative-b}
    v'(u)=\frac{E\psi(v(u)+S_{n})}{2-E\psi'(v(u)+S_{n})u}.
\end{equation}
 From this it is evident that the derivative $v'(u)$ could blow up in
 a neighborhood of a point $u$ where $E\psi'(v(u)+S_{n})u=2$. The
 following corollary shows that this cannot occur in any region where
 the function $|v(u)|$ remains suitably small. 

\begin{corollary}
  \label{corollary:big-v'-big-v}
   There exists a constant $\beta>0$ independent of $n$
   and of the particular Nash equilibrium $v(u)$ such that
  for all $u\in [\underline{u} ,-1]\cup [1,\overline{u} ]$,
  \begin{equation}
    \label{eq:big-v'-big-v}
    \limsup _{u' \rightarrow u} |v(u')|\leq \beta \quad
    \Longrightarrow \quad |v'(u)|\leq \xnorm{\psi}_{\infty}. 
  \end{equation}
\end{corollary}

\begin{proof}
  The constant $\beta$ can be chosen smaller than the discontinuity
  threshold $\Delta$
  (cf. Proposition~\ref{proposition:discontinuities-necessary}), in
  which case the inequality
  $\limsup _{u' \rightarrow u} |v(u')|\leq \beta$ implies that
  $v(\cdot)$ is continuous at $u$.  Consequently, $v(\cdot)$ is
  differentiable at $u$ unless the denominator in
  \eqref{eq:1st-derivative-b} is $0$, by
  Proposition~\ref{proposition:smoothness}. 

  Assume now that $|u|\geq 1$.  By
  the necessary condition \eqref{eq:weierstrass}, if
  $|2v(u)|\leq \beta$ then $E\psi(v(u)+S_{n})\leq \beta$, and by
  Lemma~\ref{lemma:psi-psi-prime}, if $\beta >0$ is sufficiently
  small, then $E|\psi'(v(u)+S_{n})u|< 1$, and so the denominator on
  the right side of \eqref{eq:1st-derivative-b} is greater than
  $1$. Thus, \eqref{eq:1st-derivative-b} implies that if
  $|v(u)|<\beta$ then
  \begin{displaymath}
    |v'(u)|\leq E\psi(v(u)+S_{n})\leq \xnorm{\psi}_{\infty}.
  \end{displaymath}
\end{proof}

\subsection{An \emph{A Priori} Bound on Nash Equilibria}

\begin{lemma}\label{lemma:aPrioriLB}
  For all sufficiently large $n$, in any Nash equilibrium,
\begin{equation}\label{eq:first-lower-bound}
      \max _{\underline{u}_{1/n}\leq u \leq \overline{u}_{1/n}}|v(u)| \geq \frac{\delta}{4n}.
\end{equation} 
\end{lemma}

\begin{proof}
  By Proposition~\ref{proposition:strict-monotonicity}, the maximum occurs
  at one of the two endpoints $\underline{u} _{1/n}$ or
  $\overline{u} _{1/n}$.  Suppose that the inequality
  \eqref{eq:first-lower-bound} were not true; then an agent with value
  $u \in[1,\overline{u}_{1/n}]$ would purchase at most $\delta/4n$
  votes (in absolute value). On the other hand, by the necessary
  condition \eqref{eq:weierstrass}, the number of votes $v(u)$ bought
  by this agent is  $u \times$ the agent's marginal pivotality
  $\frac{1}{2} E\psi(v(u)+S_{n})$. Consider the event that all $n$ of the
  remaining agents have values in
  $[\underline{u} _{1/n},\overline{u} _{1/n}]$: the probability of
  this event is
  \begin{align*}
    \left(
      1- \int_{\underline{u} }^{\underline{u} _{1/n}}f(t)\,dt -
      \int_{\overline{u} _{1/n}}^{\overline{u} }f(t)\,dt
    \right)^{n} 
    &\approx (1-f(\underline{u} )/n - f(\overline{u} )/n)^{n}\\
    &\approx \exp\{-f (\underline u)-f (\overline u)\}:=p>0,
  \end{align*}
  and on this event, the vote total  would not exceed $\delta/4
  +\delta/(4n)\leq \delta/2$ in absolute value. Consequently, the
  marginal pivotality satisfies
  \begin{displaymath}
    E\psi(v(u)+S_{n})\geq p \min _{x\in
      [-\delta/2,\delta/2]}|\psi(x)|, 
  \end{displaymath}
   and so
  \begin{displaymath}
    2|v(u)|\geq |u|  p E\psi(v(u)+S_{n}) \geq p\min _{x\in [-\delta/2,\delta/2]}|\psi(x)|.
  \end{displaymath}
  Since $\psi$ is bounded away from $0$ on the interval $[-\delta/2,\delta/2]$,
  for large $n$ this would contradict our supposition that
  $\max |v(u)|$ on the interval $u\in [\underline{u}
  _{1/n},\overline{u} _{1/n}]$ is below $\delta/(4n)$.
\end{proof}

\section{Weak Consensus Estimates}\label{ssec:consensus} 

According to Proposition  \ref{proposition:weierstrass}, in any Nash
equilibrium the number of votes $v (u)$ an agent with utility $u$
purchases is $u\times$ the \emph{marginal pivotality}
$\frac{1}{2} E\psi (v(u)+S_{n})$. When the sample size $N=n+1$ is large, the
effect of a single voter's contribution $v(u)$ to the vote total
$v(u)+S_{n}$ should be small, and so one expects that an agent with
value $u$ should have approximately the same marginal pivotality as
an agent with value $u'$. However, as we will show later, this is not always
true: an agent with value $u$ in the extreme tails of the distribution
$F$ can, in some circumstances, have a drastically different marginal
pivotality than agents with values in the bulk of the distribution.

Nevertheless, agents with values in the bulk of the distribution do
have nearly the same marginal pivotalities, as the following lemma
shows. This lemma will figure critically in the proof of
Lemma~\ref{proposition:concentrate} (the key anti-concentration estimate)
in section~\ref{ssec:size} below.

\begin{lemma}\label{lemma:consensusA}
For any $\epsilon >0$ and any $0<\alpha <1$,  if $n$
is sufficiently large  then in every
Nash equilibrium $v (\cdot)$,
\begin{equation}\label{eq:consensus}
  1-\alpha\leq\frac{E\psi (v (u)+S_{n}) }{E\psi (v (u')+S_{n})}\leq(1-\alpha)^{-1}
\end{equation}
for any two values $u,u'$ not within distance $\epsilon$ of either
$\underline u$, or $\overline u$, or $0$. Furthermore, there exists
$C>0$ such that  for any $\beta>0$, if $n$ is sufficiently large then
in any equilibrium, for all
$u\in[\underline{u} ,\overline{u} ]\setminus (-\epsilon,\epsilon )$, 
\begin{align}
  \label{eq:tail-consensus+}
	\frac{E\psi (v (u)+S_{n})}{E\psi (v (\overline u_{\beta
	n^{-3/2}})+S_{n})}&\geq \frac{C\beta}{n^{1/2}}
                            \quad \text{and}\\
  \label{eq:tail-consensus-}
		\frac{E\psi (v (u)+S_{n})}{E\psi (v (\underline u _{\beta
	n^{-3/2}})+S_{n})}&\geq \frac{C\beta}{n^{1/2}}.
\end{align}
\end{lemma}

\begin{proof}
  Let $J_{1},J_{2},\dotsc ,J_{k}$ be any partition of the interval
  $[\underline u ,\overline u]$ into non-overlapping Borel sets of
  positive Lebesgue measure, and for each index $i$ let $M_{i}$ be the
  number of agents (in the entire sample of size $N=n+1$) with values
  in the set $J_{i}$.  The random vector $(M_{1},M_{2},\dots ,M_{k})$
  has the multinomial distribution
\[
	P (M_{i}=m_{i} \; \text{for each} \; i\leq k)
	=\frac{(n+1)!}{m_{1}!m_{2}!\dotsb m_{k}!} \prod_{i=1}^{k} p_{i}^{m_{i}}
\]
where 
\[
	p_{i}= P (U_{1}\in J_{i})=\int_{J_{i}} f (u)\,du.
\]
Conditional on the event that $M_{i}=m_{i}$ for each $i\leq k$, the
sample $\{ U_{1},U_{2},\dotsc ,U_{n}\}$ has the same distribution as a
stratified random sample gotten by choosing $m_{i}$ elements from the
set $J_{i}$ according to the density $f\mathbf{1}_{J_{i}}/p_{i}$
for each $i\leq k$. Consequently, for any choice of index $i\leq k$, 
\begin{multline}\label{eq:stratify-cond}
	E[\psi (v (U_{n+1})+S_{n})\,|\, U_{n+1}\in J_{i}]=\\
	\sum_{m_{1},m_{2},\dotsc ,m_{k}}
	\frac{n!}{m_{1}!m_{2}!\dotsb m_{k}!} 
	\left( \prod_{i=1}^{k}
	p_{i}^{m_{i}}\right) E_{*} (m_{1},m_{2},\dotsc ,m_{i-1} ,m_{i}+1, m_{i+1},\dotsc m_{k})
\end{multline}
where
\[
	E_{*} (m_{1},m_{2},\dotsc ,m_{k})=E (\psi (S_{n+1})\,|\, M_{i}=m_{i} \; \forall \,i\leq k)
\]
for any set of nonnegative integers $m_{j}$ that sum to $n+1$. Note
that the conditional expectations $E_{*} (m_{1},m_{2},\dotsc ,m_{k})$ are all
nonnegative, and are bounded above by $\xnorm{\psi}_{\infty}$.

The proof of the lemma will be based on systematic exploitation of
equation \eqref{eq:stratify-cond}.  To relate the conditional
expectation on the left side of \eqref{eq:stratify-cond} to the marginal pivotality
$\frac{1}{2} E\psi(v(u)+S_{n})$, we appeal to the monotonicity of $v (\cdot)$. Fix
$u\in (\epsilon ,\overline u]$ and let $J=[u-\alpha u,u]$ and
$J'=[u,u+\alpha u]\cap [u, \overline{u} ]$, where $\alpha >0$ is small enough that
$J\subset [0 ,\overline u]$; then by \eqref{eq:weierstrass} and
the monotonicity of $v$,
\begin{align*}
	 (1-\alpha ) E\psi
	(v (u')+S_{n})&\leq E\psi (v (u)+S_{n})  \quad \text{for all} \; u'\in J\quad \text{and}\\
	(1+\alpha ) E\psi
	(v (u'')+S_{n})&\geq E\psi (v (u)+S_{n})\quad \text{for all} \; u''\in J'.
\end{align*}
Consequently,
\begin{align}\label{eq:deconditioning}
		(1-\alpha)E[\psi ( v (U_{n+1})+S_{n})\,|\, U_{n+1}\in
  J]&\leq E\psi (v (u)+S_{n})\quad \text{and}\\ 
\notag 	(1+\alpha)E[\psi ( v (U_{n+1})+S_{n})\,|\, U_{n+1}\in J']&\geq E\psi (v (u)+S_{n}).
\end{align}
A similar argument shows that for values
$u\in[\underline{u} , -\epsilon]$, the inequalities
\eqref{eq:deconditioning} hold with $J=[u,u-\alpha u]$ and
$J'=[u+u\alpha ,u]\cap[\underline{u} ,u]$.  Thus, to prove the
inequalities \eqref{eq:consensus}, \eqref{eq:tail-consensus+}, and
\eqref{eq:tail-consensus-}, it suffices to prove analogous
inequalities for conditional expectations of the form
\eqref{eq:stratify-cond}.

\begin{proof}
  [Proof of inequalities \eqref{eq:tail-consensus+} and
  \eqref{eq:tail-consensus-}]
  There are two cases for each inequality, depending on whether $u\geq
  \epsilon$ or $u\leq -\epsilon$; by symmetry (i.e., reversing the
  roles of \emph{positive} and \emph{negative} values and votes), it
  is enough to consider only the case $u \geq \epsilon$. Moreover, it
  suffices to prove only one of the inequalities \eqref{eq:tail-consensus+} and
  \eqref{eq:tail-consensus-}, specifically, the one for which the
  denominator on the left side is larger, because the other will then
  follow trivially. For the sake of exposition, let's assume that
  \begin{displaymath}
    E\psi (v (\overline u_{\beta
	n^{-3/2}} +S_{n})\geq E\psi (v (\underline u _{\beta
	n^{-3/2}}+S_{n});
    \end{displaymath}
    the other case is similar. Under this assumption,
    \begin{equation}\label{eq:big-denominator}
      E\psi (v (\overline u_{\beta
	n^{-3/2}}+S_{n} )\geq \gamma n^{-1}
    \end{equation}
    for some $\gamma >0$ independent of $n$, by the \emph{a priori}
    bound of Lemma~\ref{lemma:aPrioriLB}. 
    
  Fix 
  $\beta>0$ and $\alpha\in (0,1)$,
  and let
  \begin{align*}
    J_{1}&=[u-\alpha u,u],\\
    J_{2}&=[\overline{u} _{\beta n^{-3/2}},\overline{u} ], \quad
           \textrm {and}\\
    J_{3}&= [\underline{u} ,\overline{u} ]\setminus (J_{1}\cup J_{2}).
  \end{align*}
  By 
  \eqref{eq:deconditioning}, to prove \eqref{eq:tail-consensus+} it
  suffices to show that 
  \begin{equation}\label{eq:tail-comparability}
	\frac{E\psi ( v (U_{n+1})+S_{n})\,|\, U_{n+1}\in J_{1}]}{E\psi
	( v (U_{n+1})+S_{n})\,|\, U_{n+1}\in J_{2}]} \geq \frac{C'\beta}{n^{1/2}},
    \end{equation}
    for some constant $C'>0$ independent of $n$ and $\beta$.
    Now \eqref{eq:big-denominator}, together with inequality
    \eqref{eq:deconditioning} and the monotonicity of $v$, implies
    that the denominator is of size at least
    \begin{displaymath}
      E\psi ( v (U_{n+1})+S_{n})\,|\, U_{n+1}\in J_{2}]\geq \gamma ' n^{-1};
    \end{displaymath}
    consequently, in proving \eqref{eq:tail-comparability} we can
    ignore errors that decay exponentially in $n$.

    For brevity, denote the numerator on the left side of
    \eqref{eq:tail-comparability} by $E_{1}$ and the denominator
    by $E_{2}$. Both $E_{1}$ and $E_{2}$ can be expressed as sums of
    the form \eqref{eq:stratify-cond} (with $k=3$). For any triple
    $(m_{1},m_{2},m_{3})$ with $m_{1}\geq 1$ and $m_{2}\geq1$, terms with factor $E_{*}
    (m_{1},m_{2},m_{3})$ occur in both sums, but with different coefficients
    \begin{align*}
      		&\frac{n!}{( m_{1}-1)!m_{2}!m_{3}!}
		p_{1}^{m_{1}-1}p_{2}^{m_{2}}p_{3}^{m_{3}} \quad
		\text{in}\; E_{1}, \quad \text{and}\\
	&	\frac{n!}{m_{1}!( m_{2}-1)!m_{3}!}
		p_{1}^{m_{1}}p_{2}^{m_{2}-1}p_{3}^{m_{3}}	 \quad 
		\text{in}\;E_2,
    \end{align*}
    where $p_{1},p_{2},p_{3}$ are the $F-$probabilities of
    $J_{1},J_{2},J_{3}$, respectively. The ratio of these coefficients
    ($E_{1}$ to $E_{2}$) is
    \begin{displaymath}
      \frac{m_{1}}{m_{2}}\frac{p_{2}}{p_{1}}\sim
	\frac{m_{1}\beta f (\overline u)}{p_{1}m_{2}n^{3/2}}.
      \end{displaymath}
      We will show that  triples $(m_{1},m_{2},m_{3})$ with either
      $m_{1}\leq np_{1}/2$ or $m_{2}\geq 5$ can be ignored; for all
      other triples, the coefficient ratio is, for large $n$, at least
      \begin{displaymath}
        \frac{m_{1}}{m_{2}}\frac{p_{2}}{p_{1}}\geq
        \frac{\beta f(\overline{u} )}{(2\times 2\times 5)n^{1/2}}
      \end{displaymath}
      This will prove that inequality \eqref{eq:tail-comparability}
      holds with $C'= f(\overline{u})/(40 n^{1/2})$, provided $n$ is
      sufficiently large.

      Consider first those triples $(m_{1},m_{2},m_{3})$ with
      $m_{1}\leq np_{1}/2$. In both sums $E_{1}$ and $E_{2}$, the
      total contribution of the terms with factors
      $E_{*}(m_{1},m_2,m_3 )$ where $m_{1}\leq np_{1}/2$ is at most
      $\xnorm{\psi}_{\infty}$ times the probability that a random
      sample of size $n$ from the distribution $F$ has fewer than
      $1+np_{1}/2$ elements in the interval $J_{1}$. By Hoeffding's
      inequality (cf. Appendix~\ref{sec:hoeffding}), this probability
      is decays exponentially in $n$; therefore, the total
      contribution of these terms to either $E_{i}$ is negligible
      compared to $E_{i}$ when $n$ is large.

      Finally, consider those triples with $m_{2}\geq 5$. In  $E_{1}$,
      these triples correspond to samples with at least $5$ elements
      in the interval $J_{2}$, while in $E_{2}$ they correspond to
      samples with at least $4$ entries in $J_{2}$. The probability
      that a random sample of size $n$ contains at least $k$ elements
      in an interval of length $\beta n^{-3/2}$  is no larger than
      \begin{displaymath}
        \binom{n}{k} (\beta n^{-3/2}\xnorm{f}_{\infty})^{k}.
      \end{displaymath}
      For $k=4$ this is $O(n^{-2})$, which is of smaller order of
      magnitude than $E_{2}$ (which is $\geq \gamma'/n$, as noted
      above); and for $k=5$ it is $O(n^{-5/2})$, which is of smaller
      order than $E_{1}$, as this (by the argument above) is of size
      at least $C'E_{2}/(40 n^{1/2})$.

    \end{proof}

    \begin{proof}[Proof of inequality \eqref{eq:consensus}]
      The proof of inequality \eqref{eq:consensus} follows a similar
      line. Clearly, it is enough to prove the lower bound in
      \eqref{eq:consensus}, because the upper bound will then follow
      by reversing the roles of $u$ and $u'$.  By
      inequalities \eqref{eq:deconditioning}, it suffices to show that
      for any $\alpha >0$ and any two fixed, non-overlapping intervals
      $J_{1},J_{2}$ of positive length contained in
      $[\underline{u} _{\epsilon}, -\epsilon]\cup[\epsilon ,\overline u -\epsilon]$,
      \begin{displaymath}
        \frac{E[\psi ( v (U_{n+1})+S_{n})\,|\, U_{n+1}\in
	J_{1}]}{E[\psi ( v (U_{n+1})+S_{n})\,|\, U_{n+1}\in J_{2}]}
	\geq 1-\alpha 
      \end{displaymath}
      provided $n$ is sufficiently large. In proving this, we can
      ignore terms of exponentially small size, because
      \eqref{eq:tail-consensus+} and \eqref{eq:tail-consensus-},
      together with Lemma~\ref{lemma:aPrioriLB}, imply that both
      numerator and denominator are of magnitude at least
      (constant)$\times n^{-3/2}$. 

      As in the proof of inequality \eqref{eq:tail-consensus+}, denote
      the numerator and denominator by $E_{1}$ and $E_{2}$,
      respectively.  Each of the conditional expectations
      $E_{1},E_{2}$ has a representation \eqref{eq:stratify-cond} with
      $k=3$, where $J_{3}$ is the complement of $J_{1}\cup J_{2}$ in
      $[\underline u ,\overline u]$. For any $\gamma >0$, the
      contribution to either $E_{1}$ or $E_{2}$ from terms of
      \eqref{eq:stratify-cond} for which
      $|m_{i}-np_{i}|\geq n\gamma p_{i}$, for either $i=1$ or $i=2$,
      is exponentially small, by Hoeffding's inequality
      (cf. Appendix~\ref{sec:hoeffding}), and hence can be
      ignored. For any triple $(m_{1},m_{2},m_{3})$ with $m_{i}\geq 1$
      and $m_{2}\geq 1$, terms with factor $E_{*} (m_{1},m_{2},m_{3})$
      occur in both $E_{1}$ and $E_{2}$, with the same coefficients as
      in the proof of \eqref{eq:tail-consensus+}.  The ratio of these
      coefficients ($E_{1}$ to $E_{2}$) is
      \begin{displaymath}
        \frac{m_{1}}{m_{2}} \frac{p_{2}}{p_{1}}.
      \end{displaymath}
      Since only those triples with $|m_{i}-np_{i}|< n\gamma p_{i}$
      contribute substantially to the expectations, it follows that
      for large $n$,
      \begin{displaymath}
        \frac{E_{1}}{E_{2}}\geq \frac{1+2\gamma}{1-2\gamma }.
      \end{displaymath}
      Clearly, if $\gamma >0$ is sufficiently small then the lower bound will
      exceed $1-\alpha$.
      \end{proof}

\end{proof}

\section{Concentration and size constraints}\label{ssec:size} Because
the vote total $S_{n}$ is the sum of independent, identically
distributed random variables $v (U_{i})$ (albeit with unknown
distribution), its distribution is subject to concentration
restrictions, such as those imposed by the following lemma.

\begin{lemma}\label{proposition:anti-concentration}
For any $\epsilon >0$ and any $\beta>0$ there exists $\gamma<\infty$
such that for all sufficiently large values of $n $ and any Nash
equilibrium $v (u)$, if
\begin{equation}\label{eq:v-big-norm}
	\xnorm{v}_{\infty}\geq \epsilon 
\end{equation}
then
\begin{equation}\label{eq:mid-range-bound}
	|v (u)|\leq \frac{\gamma}{\sqrt{n}} \quad \text{for all} \;
	u\in [\underline u_{\beta},\overline u_{\beta}].
\end{equation}
\end{lemma}

We will deduce Lemma~\ref{proposition:anti-concentration} from the following
general fact about sums of independent, identically distributed random
variables.

\begin{proposition}\label{proposition:llt}
  Fix $\alpha >0$.  For any $\epsilon >0$ and any $C<\infty$ there
  exist constants $C'>0$ and $n'<\infty$ such that the following
  statement is true: if $n\geq n'$ and $Y_{1},Y_{2},\dotsc ,Y_{n}$ are
  independent, identically distributed random variables such that
\begin{equation}\label{eq:var-to-third}
	E|Y_{1}-EY_{1}|^{3}\leq C\text{\rm var} (Y_{1})^{3/2}
	\quad \text{and} \quad \text{\rm var} (Y_{1}) \geq C'/n
\end{equation}
then for every interval $J\subset \zz{R}$ of length $\alpha $ or greater, 
\begin{equation}\label{eq:small-local-prob}
	P\left\{\sum_{i=1}^{n}Y_{i}\in J \right\}\leq \epsilon |J|/\alpha  .
\end{equation}
\end{proposition}

The proof of this proposition, a routine exercise in the use of the
Berry-Esseen theorem, is relegated to  Appendix~\ref{sec:berry-esseen}.

\begin{proof}
  [Proof of Lemma~\ref{proposition:anti-concentration}]
  By the monotonicity of Nash equilibria and 
  the necessary condition \eqref{eq:weierstrass}, if a Nash
  equilibrium $v(u)$ has supremum norm $\xnorm{v}_{\infty}>\epsilon$
  then for at least one of the endpoints $u= \underline u$ or $u=\overline u$,
  \begin{equation}\label{eq:big-concentration}
    |u| E\psi (v(u)+S_{n}) >\epsilon ,
  \end{equation}
  and so
  \begin{displaymath}
    P \xset {S_{n}+v(u)\in [-\delta,\delta]}\geq  \frac{2\epsilon}{|u|
      \xnorm{\psi}_{\infty}} \geq \epsilon':= \min_{u=\underline u
      ,\overline u }\frac{2\epsilon}{|u| \xnorm{\psi}_{\infty}}.
  \end{displaymath}
    We will use Proposition~\ref{proposition:llt} to show that such a
    high concentration of probability in an interval of length
    $\alpha=2\delta$ is impossible unless the function $v$ is bounded
    above by $\gamma/\sqrt{n}$ in the interval $[\underline u
    _{\beta},\overline u _{\beta}]$.

    We can assume without loss of
    generality that $\beta >0$ is sufficiently small  that the
    interval $[-2\beta,2\beta]$ is contained in the open interval
    $(\underline{u} _{\beta},\overline{u} _{\beta})$, because the
    condition \eqref{eq:mid-range-bound} becomes less stringent as
    $\beta$ increases.     
    For such $\beta$, the intervals $[\underline{u}
    _{\beta},-2\beta]$ and $[2\beta, \overline{u} _{\beta}]$ both have
    positive length; hence, since the
    density $f$ is bounded below, it follows that for suitable
    constants $C>0$ and $\frac{1}{2} >p>0$ there
    are intervals
    \begin{align*}
      J_{+}&=[u_{+},\overline     u_{\beta}]\subset    (2\beta,\overline
             u   ) \quad \text{and}\\
      J_{-}&=[u_{-},\underline u_{\beta}]\subset (\underline u
             ,-2\beta),
    \end{align*}
    both with $F-$probability $p$.

      \bigskip
      \noindent
      \textbf{Claim:} If $U$ has density $f$, then for a suitable
      constant $C'<\infty$ the \emph{conditional} distribution of
      $v(U)$ given the event $G:=\xset {U\in J_{+}\cup J_{-}}$
      satisfies the moment conditions
      \begin{gather}
        \label{eq:check2}
        E(|v(U)-E(v(U)\,|\, G) |^{2}\,|\,G) \geq \frac{1}{2}\min
        (v(u_{+})^{2},v(u_{-})^{2}) \quad \textrm {and}\\
        \label{eq:check3}
        E(|v(U)-E(v(U)\,|\, G) |^{3}\,|\,
        G) \leq C' E(|v(U)-E(v(U)\,|\, G) |^{2}\,|\,G) ^{3/2}
      \end{gather}

      \begin{proof}[Proof of the Claim]
        By Lemma~\ref{proposition:no-zeros}, any Nash
        equilibrium $v(u)$ is nondecreasing, strictly 
        positive for $u>0$, and strictly negative for
        $u<0$. Consequently, $v(U)\geq v(u_{+})>0$ on the event
        $G\cap \xset {U>0}$, and $v(U)\leq v(u_{-})<0$ on the event
        $G\cap \xset {U\leq 0}$.  By construction, the conditional
        probability that $U>0$ given the event $G$ is
        $p/(2p)=1/2$. Thus,
        \begin{align*}
          E(|v(U)-E(v(U)\,|\, G) |^{2}\,|\,G) &\geq E(|v(U)
                                                |^{2}
                                                \mathbf{1}_{\xset {U
                                                >0}}\,|\,G)\geq
                                                \frac{1}{2} 
                                                v(u_{+})^{2} \quad
                                                \textrm {if} \;
                                                E(v(U)\,|\, G)\leq
                                                0,\\
          E(|v(U)-E(v(U)\,|\, G) |^{2}\,|\,G) &\geq E(|v(U)
                                                |^{2}
                                                \mathbf{1}_{\xset {U
                                                \leq 0}}\,|\,G)\geq
                                                \frac{1}{2} 
                                                v(u_{-})^{2} \quad
                                                \textrm {if} \;
                                                E(v(U)\,|\, G)\geq 0.                                                
        \end{align*}
        This proves \eqref{eq:check2}.

        Denote by $v_{*}=\min (v(u_{+}),|v(u_{-})|)>0$ the minimum
        value of $|v(u)|$ on the set $u\in J_{+}\cup J_{-}$.
        By Lemma~\ref{lemma:consensusA}, there is a
        constant $A>0$ depending on $\beta >0$ but not on $n$ such
        that for any Nash equilibrium,
    \begin{displaymath}
      v_{*}:= \min_{ u\in J_{+}\cup J_{-}.}|v(u)| \geq A \max _{u\in J_{+}\cup J_{-}}|v(u)|.
    \end{displaymath}
    Clearly, the conditional expectation $E(v(U)\,|\, G)$ falls in the
    interval $[-v_{*},v_{*}]$, so deviations from the conditional
    expectation are bounded by $2 v_{*}$. Therefore,
    \begin{displaymath}
      E(|v(U)-E(v(U)\,|\, G) |^{3}\,|\,
        G) \leq 8  \max _{u\in J_{+}\cup J_{-}}|v(u)|^{3} \leq 8 (v_{*}/A)^{3}.
      \end{displaymath}
      But \eqref{eq:check2} implies that  the conditional variance of
      $v(U)$ is at least $\frac{1}{2} v_{*}^{2}$, so \eqref{eq:check3}
      follows, with $C'= 8 \cdot 2^{3/2} \cdot A^{-2}$.
      \end{proof}

      Let $M$ be the number of points $U_{i}$ in the sample
      $U_{1},U_{2},\dotsc ,U_{n}$ that fall in $J_{+}\cup J_{-}$, and
      let $S^{*}_{n}$ be the sum of the votes $v (U_{i})$ for those
      agents $i$ whose values $U_{i}$ fall in this range. By
      construction, $M$ has the binomial-$(n,2p)$ distribution.
      Moreover, conditional on the event $M=m$ and
      $S_{n}-S^{*}_{n}=w$, the random variable $S^{*}_{n}$ is the sum
      of $m$ independent random variables $Y_{i}$ whose common
      distribution is the conditional distribution of $v(U)$ given the
      event $G=\xset {U\in J_{+}\cup J_{-}}$. By the Claim above, the third
      moment hypothesis of Proposition~\ref{proposition:llt} is met,
      so for any $\epsilon' >0$ there exists
      $C''=C''(\epsilon')<\infty$ such that if
      \begin{equation}\label{eq:too-big}
        \max (v(\overline{u} _{\beta}),|v(\underline{u}
        _{\beta})|)\geq C''/\sqrt{n}
      \end{equation}
       then the conditional probability, given
      $M=m\geq n p$ and $S_{n}-S^{*}_{n}=w$, that $S^{*}_{n}$ lies in
      any interval of length $4\delta$ is bounded above by
      $\epsilon'$. This in turn implies 
\[
	E\psi (S_{n}+v)\leq \xnorm{\psi}_{\infty}\left( P\{M\leq np \}
	+\epsilon '\right)  \quad \textrm {for all} \; v\in \mathbb{R}.
\]
Since $P\{M\leq n p \}\rightarrow 0$ as $n \rightarrow \infty$, by the
weak law of large numbers, it then follows that for all sufficiently
large $n$,
\begin{displaymath}
  E\psi (S_{n}+v)\leq 2\epsilon' \xnorm{\psi}_{\infty} .
\end{displaymath}
For small $\epsilon'$, this bound is smaller than $\epsilon/\max
(\overline{u} _{\beta}, |\underline{u} _{\beta}|)$, so the hypothesis
\eqref{eq:too-big} is incompatiable with \eqref{eq:big-concentration}. 

\end{proof}

Lemma~\ref{proposition:anti-concentration} implies that for any
$\beta >0$, if $n$ is sufficiently large then for any Nash
equilibrium $v (u)$, the absolute value $|v (u)|$ must be small except
at  values $u$ within distance $\beta $ of one of
the endpoints $\underline u ,\overline u$. The following lemma
improves this bound to the extreme tails of the distribution.

\begin{lemma}\label{proposition:size}
For any $0<\epsilon<\infty $ there exists $\alpha >0$ such that for
all sufficiently large $n$, every Nash equilibrium $v (u)$ satisfies
the inequality
\begin{equation}\label{eq:size}
	|v (u)|\leq \epsilon 
\end{equation}
for all $u$ at distance greater than $\alpha n^{-3/2}$ from both
endpoints $\underline u ,\overline u$. 
\end{lemma}

\begin{proof}
  It suffices, by symmetry, to consider only values $u >0$.
  Lemma~\ref{proposition:anti-concentration} implies that for any
  $\epsilon >0$ there exists $\gamma >0$ such that if $n$ is
  sufficiently large and $\xnorm{v}_{\infty}\geq \epsilon$ then
  $2|v (u)|\leq \gamma /\sqrt{n}$ for all
  $u\in [\underline u_{\epsilon},\overline u_{\epsilon}]$. Hence, by
  the necessary condition \eqref{eq:weierstrass},
  \begin{displaymath}
    E\psi (v (\overline u _{\epsilon })+S_{n})\overline u _{\epsilon }\leq \frac{\gamma
	}{\sqrt{n}}.
      \end{displaymath}
      But by Lemma~\ref{lemma:consensusA} (cf. inequalities
      \eqref{eq:tail-consensus+} and \eqref{eq:tail-consensus-}),
      there exists $C>0$ such that for any $\alpha >0$, if $n$ is
      sufficiently large then for every Nash equilibrium $v$,
      \begin{displaymath}
        	\frac{E\psi (v (u)+S_{n})}{E\psi (v (\overline u_{\alpha
	n^{-3/2}})+S_{n})}\geq \frac{C\alpha}{\sqrt{n}} 
  \end{displaymath}
  for all $u\in[\underline{u} _{\epsilon},\overline{u} _{\epsilon}]$,
  and in particular for $u=\overline u_{\epsilon }$.  The last two
  displayed inequalities now combine to yield
  \begin{displaymath}
    E\psi (v (\overline u_{\alpha n^{-3/2}})+S_{n})\leq
  \frac{\gamma}{C\alpha}
	\quad \Longrightarrow \quad 2v (\overline u_{\alpha
          n^{-3/2}})\leq \frac{\gamma \overline{u} }{ C\alpha}.
      \end{displaymath}
      Thus, the inequality $2v (u)>\epsilon$ can hold at some
      $u = u_{\alpha n^{-3/2}}$ only if
      $\alpha < (\gamma \overline{u} )/(C \epsilon)$.
\end{proof}

\begin{proposition}
  \label{proposition:max-cont-interval}
  For any Nash equilibrium $v(u)$ there is a maximal nonempty interval
  $J_{\max }$ containing $u=0$ in its interior on which $v(u)$ is
  continuous.  For sufficiently large $n$ the endpoints
  $u_{-}<0<u_{+}$ of this interval lie within distance
  $\beta n^{-3/2}$ of $\underline{u} $ and $\overline{u} $,
  respectively, for some constant $\beta>0$ not depending on $n$ or on
  the particular Nash equilibrium.  The function $v(u)$ is
  $C^{\infty}$ on the interior of $J_{\max}$, and for any
  $\epsilon>0$, if $n$ is sufficiently large then
  \begin{equation}
    \label{eq:max-cont-interval}
    \sup_{u \in J_{\max}}|v(u)|<\epsilon.
  \end{equation}
  Consequently, for  any $\epsilon>0$ there exists  $n_{\epsilon}\in
  \mathbb{N}$ such 
  that if $n \geq n_{\epsilon}$ then any Nash equilibrium $v(u)$ with
  no discontinuities satisfies
  \begin{equation}
    \label{eq:continuous-equilibria-small}
    \xnorm{v}_{\infty}<\epsilon.
  \end{equation}
\end{proposition}

\begin{proof}
  The existence of a nonempty interval of continuity containing $u=0$
  is clear, because any Nash equilibrium is continuous at $0$. That
  the endpoints $u_{-}<u_{+}$ of this interval lie within distance
  $\beta n^{-3/2}$ of the endpoints $\underline{u} , \overline{u} $
  follows from Proposition~\ref{proposition:discontinuities-necessary}
  and Lemma~\ref{proposition:size}, because the former asserts
  that there is no discontinuity on any interval in which
  $|v(u)|<\Delta$, where $\Delta>0 $ is the discontinuity threshold,
  and the latter implies that $|v(u)|$ is bounded by $\Delta$ on the
  interval
  $ (\underline{u} _{\beta n^{-3/2}}, \overline{u} _{\beta
    n^{-3/2}})$, for some $\beta >0$ independent of $n$ and the
  particular Nash equilibrium.

  It remains to show that $v(u)$ is not only continuous but
  $C^{\infty}$ on the interior of $J$, and that for sufficiently large
  $n$ the inequality \eqref{eq:max-cont-interval} holds.  Since $v(u)$
  is continuous on $J$, Proposition~\ref{proposition:smoothness}
  implies that $v(u)$ fails to be $C^{\infty}$ only at points $u$
  where the equality \eqref{eq:bad-u} holds.  But
  Lemma~\ref{lemma:psi-psi-prime} and the necessary condition
  \eqref{eq:weierstrass} imply that if $\epsilon>0$ is 
  sufficiently small then equation \eqref{eq:bad-u}  cannot be
  satisfied at any point where $2|v(u)|<\epsilon$. Thus, to complete
  the proof it     will suffice to show that for any $\epsilon>0$, if
  $n$ is sufficiently large then $|v(u)|<2\epsilon$ on the interval $J_{\max}$.

  Assume that $2\epsilon <\Delta$, where $\Delta$ is the discontinuity
  threshold. By 
  Corollary~\ref{corollary:big-v'-big-v}, if $\epsilon>0$ is
  sufficiently small then $|v'(u)|\leq \xnorm{\psi}_{\infty}$ for all $u \geq 1$ at which
  $|v(u)|<\epsilon$. By Proposition~\ref{proposition:size}, the
  inequality 
   $|v(u)|<\epsilon$ holds for all $u$ at distance more
  than $\beta n^{-3/2}$ from the endpoints
  $\underline{u} ,\overline{u} $, where $\beta>0$ does not depend on
  $n$ or on the particular Nash equilibrium.  Now suppose that
  $v(u)\geq 2\epsilon$ for some $u \in [1,\overline{u} ]\cap J_{\max}$, and define
  \begin{displaymath}
    u_{*}:= \inf \xset {u \geq 1\,:\, v(u)\geq 2\epsilon}.
  \end{displaymath}
  Since $v(u)$ has no discontinuities in $J_{\max}$, it must be the
  case that $v(u_{*})=2 \epsilon$.  But the fundamental theorem of
  calculus implies
  \begin{align*}
    v(u_{*})&=v(\overline{u} _{\beta n^{-3/2}})+\int_{\overline{u}
      _{\beta n^{-3/2}}}^{u_{*}}v'(u)\,du \\
    &\leq \epsilon + (u_{*}-\overline{u} _{\beta n^{-3/2}})\xnorm{\psi}_{\infty}\\
    &\leq \epsilon + \beta n^{-3/2}\xnorm{\psi}_{\infty};
  \end{align*}
  thus, if $n$ is sufficiently large then the equality
  $v(u_{*})=2\epsilon$ is impossible. The same argument shows that for
  large $n$
  there can be no $u_{*}\in [\underline{u} ,-1]\cap J_{\max}$ at which
  $v(u_{*})=-2\epsilon$. 
\end{proof}

\section{Approximate Proportionality}\label{sec:app-prop}

\subsection{The approximate proportionality
rule}\label{ssec:linearity}   We
have shown in Lemma~\ref{proposition:size} that any Nash equilibrium
$v (u)$ must be small (in magnitude) except in the extreme tails of
the distribution (in particular, for all $u$ at distance much more
than $n^{-3/2}$ from both endpoints $\underline u ,\overline
u$). Because $\psi$ is uniformly continuous, it follows that the
marginal pivotality $\frac{1}{2} E\psi (v (u)+S_{n})$ cannot differ by
very much from $\frac{1}{2} E\psi (S_{n})$. Thus, the approximation
$2v(u)\approx E\psi (S_{n})u$ is valid up to an error of size
$\epsilon_{n }|u|$ where $\epsilon_{n}\rightarrow 0$ as
$n \rightarrow \infty$. However, as $n \rightarrow \infty$ the
expectation $E\psi (S_{n})\rightarrow 0$, and so the error in the
approximation above might be considerably larger than the
approximation itself. The following proposition makes the stronger
assertion that when $n$ is large the relative error in the approximate
proportionality rule is small. This extends the range of validity of
the weak consensus estimate \eqref{eq:consensus} to the maximal
interval $J_{\max}$ of continuity.

\begin{proposition}\label{proposition:proportionality}
  For all $n$ sufficiently large, in any Nash equilibrium,
  $E\psi(S_{n})>0$, and for any $\epsilon >0$ there exists
  $n_{\epsilon}$ such that for all $n\geq n_{\epsilon}$ and any Nash
  equilibrium $v(u)$, if $J_{\max}$ is the maximal interval of
  continuity of $v(u)$ containing $u=0$ then
\begin{equation}\label{eq:proportionality}
	(1+\epsilon )^{-1}\leq \frac{2v(u)}{E\psi (S_{n}) u}
        \leq (1+\epsilon ) \quad \textrm {for all} \; u \in
      J_{\max}\setminus \xset{0}.
\end{equation}
\end{proposition}

\begin{proof} [Proof of Proposition~\ref{proposition:proportionality}]
  Since $\psi$ is $C^{\infty}$ and has compact support
  $[-\delta,\delta]$, the function $v\mapsto E\psi (v+S_{n})$ is also
  $C^{\infty}$, with derivative $E\psi '(v+S_{n})$. By
  Proposition~\ref{proposition:max-cont-interval}, the function $v(u)$
  is $C^{\infty}$ on the interval $J_{\max}$, and so, by the chain
  rule, the function $u\mapsto E\psi(v(u)+S_{n})$ is $C^{\infty}$,
  with derivative $E\psi'(v(u)+S_{n})v'(u)$. Hence, by Taylor's
  theorem, for every $u\in J_{\max}$ there exists $\tilde{u}(u)$
  intermediate between $0$ and $u$ such that
  \begin{gather}
    \label{eq:consensus-B}
    \notag 
    2v (u)=E\psi (v (u)+S_{n})u=E\psi (S_{n})u +E\psi'
    (v(\tilde{u}(u))+S_{n})v (u) u \quad \Longrightarrow \\
    (2-E\psi'(v(\tilde{u}(u))S_{n})u)v(u)=E\psi (S_{n})u .
  \end{gather}
  We will prove that for any $\epsilon'>0$, if $n$ is sufficiently
  large then in any Nash equilibrium $v(u)$,
  \begin{equation}
    \label{eq:Epsi-prime}
    |E\psi'(v(\tilde{u}(u))+S_{n})u|<\epsilon' \quad \textrm {for all}
    \; u \in J_{\max}.
  \end{equation}
  This will imply  that $E\psi(S_{n})\not=0$ (because otherwise
  equation \eqref{eq:consensus-B} would imply that $v(u)=0$ for all $u
  \in J_{\max}$, contradicting
  Proposition~\ref{proposition:strict-monotonicity}), and 
  the assertion \eqref{eq:proportionality} will then follow directly
  from \eqref{eq:consensus-B}.

  By Lemma~\ref{lemma:psi-psi-prime}, for any $\epsilon'>0$ there
  exists $\varrho=\varrho(\epsilon')>0$ such that if
  $E\psi(\tilde{v}+S_{n})|u|<\varrho$ then
  $E|\psi'(\tilde{v}+S_{n})u|<\epsilon'$. Thus, to prove
  \eqref{eq:Epsi-prime} it suffices to show that for any $\varrho >0$,
  if $n$ is sufficiently large then in every  Nash equilibrium,
  \begin{equation}\label{eq:rhoPrime}
    E\psi(v(\tilde{u}(u))+S_{n})|u| <\varrho \quad \textrm {for all} \;
    u\in J_{\max}.
  \end{equation}
  Since $|u|$ is bounded by $\max (|\underline{u} |,\overline{u} )$
  and $\tilde{u}(u)$ is intermediate between $0$ and $u$, to prove
  \eqref{eq:rhoPrime} it is enough  to show that for any $\varrho' >0$,
  if $n$ is sufficiently large then in every Nash equilibrium,
  \begin{equation}
    \label{eq:margPivotalitySmall}
    E\psi(v(u)+S_{n})<\varrho ' \quad \textrm {for all} \;
    u\in J_{\max}.
  \end{equation}
  
  Proposition~\ref{proposition:max-cont-interval} implies that for
  any $\varrho ''>0$ the function $2|v(u)|$ is bounded above by
  $\varrho ''$ on the interval $J_{\max}$, provided $n$ is 
  large. Hence, by the necessary condition~\eqref{eq:weierstrass},
  \begin{equation}
    \label{eq:rho''}
    E\psi(v(u)+S_{n})<\varrho '' \quad \textrm {for all} \;
    u\in J_{\max}\setminus  (-1,1).
  \end{equation}
  Thus, inequality \eqref{eq:margPivotalitySmall} will hold outside
  the interval $(-1,1)$ if $\varrho ''< \varrho '$. On
  the other hand,
   by the fundamental theorem of
   calculus, for any $u $ in the interval $(-1, 1)$,
   \begin{align*}
     |E\psi (v(u)+S_{n})-E\psi(v(-\alpha)+S_{n})|&\leq\int_{v(-1
       )}^{v(u)}E|\psi'(v(t)+S_{n})| \, dt \\
                                                 &\leq \xnorm{\psi'}_{\infty} (v(1)-v(-1)) \\
                                                 &\leq 2\xnorm{\psi'}_{\infty} \max _{u\in J_{\max}}|v(u)| \\
     &\leq \xnorm{\psi'}_{\infty} \varrho '';
   \end{align*}
   consequently, if $\xnorm{\psi'}_{\infty} \varrho ''<\varrho '$ then
   \eqref{eq:margPivotalitySmall} will hold in the interval $(-1,1)$. 
\end{proof}

\subsection{Consequences}\label{ssec:consequences}

Proposition~\ref{proposition:proportionality} implies that any agent with
value $u$ in the maximal interval $J_{\max}$ of continuity of a Nash
equilibrium $v(u)$ will buy $E\psi(S_{n})u$ votes, up to an error of
size $o(v(u))$. Henceforth, we will refer to agents with values
$u\in J_{\max}$ as \emph{moderates}, and agents with values
$u\not\in J_{\max}$ as \emph{extremists}. By
Proposition~\ref{proposition:max-cont-interval}, there is a constant
$\beta$ such that the endpoints of $J_{\max}$ are within distance
$\beta n^{-3/2}$ of the endpoints $\underline{u} ,\overline{u} $ of
the support of the sampling density $f$; since $f$ is continuous at
the endpoints $\underline{u} ,\overline{u} $, it follows that when $n$ is
large, the extremist region
$[\underline{u} ,\overline{u} ]\setminus J_{\max}$ has $F-$probability
less than $2\beta n^{-3/2}(f(\underline{u} )+f(\overline{u}
))$. Consequently, if
\begin{equation}\label{eq:no-extremists}
  G=\bigcap_{i=1}^{n}\xset {U_{i}\in J_{\max}}
\end{equation}
is the event that the all-but-one sample contains no extremists, then for all
sufficiently large $n$,
\begin{equation}\label{eq:probability-no-extremists}
  P(G^{c})\leq 2\beta n^{-1/2} (f(\underline{u})+f(\overline{u})).
\end{equation}
In particular, the probability that the sample contains an extremist
is  vanishingly small as $n \rightarrow\infty$.

Given the event $G$, the conditional distribution of the sample
$U_{1},U_{2}, \cdots ,U_{n}$ is the same as the \emph{unconditional}
distribution of a random sample of size $n$ from the density
\begin{displaymath}
  f_{G}(u):= \frac{f(u)\mathbf{1}_{J_{\max}}}{\int_{J_{\max}}f}.
\end{displaymath}
Since voters in the moderate range follow the approximate
proportionality rule \eqref{eq:proportionality}, it follows that for
large $n$, conditional on the event $G$, the votes $v(U_{i})$ are
independent, identically distributed random variables bounded above
and below by $2E\psi(S_{n})\underline{u} $ and
$2E\psi (S_{n})\overline{u}$. Thus, Hoeffding's inequality
(cf. Appendix~\ref{sec:hoeffding}) implies the following.

\begin{corollary}\label{corollary:hoeffding}
For all sufficiently large $n$ and any Nash
equilibrium $v (u)$,
\begin{equation}\label{eq:hoeffding-conditional}
	P (|S_{n}-E (S_{n}|G)|\geq t \,|\, G)\leq 2\exp
	\{-2t^{2}/(4n [E\psi (S_{n})]^{2} (\overline u -\underline u)^{2})\}; 
\end{equation}
and so for any Nash equilibrium with no discontinuities,
\begin{equation}\label{eq:hoeffding-no-discontinuities}
	P (|S_{n}-ES_{n}|\geq t ) \leq 2\exp
	\{-2t^{2}/(4n [E\psi (S_{n})]^{2} (\overline u -\underline u)^{2})\}. 
\end{equation}
\end{corollary}

Proposition~\ref{proposition:proportionality} also implies uniformity in
the normal approximation to the distribution of $S_{n}$, because the
proportionality rule \eqref{eq:proportionality} guarantees that the
ratio of the (conditional on $G$) third moment to the $3/2$ power of
the (conditional) variance of $v (U_{i})$ is uniformly bounded. Hence,
by the Berry-Esseen theorem (cf. Appendix~\ref{sec:berry-esseen}), we
have the following corollary.

\begin{corollary}\label{corollary:berry-esseen}
There exists $\kappa <\infty$  such that  for all sufficiently large
$n$ and any Nash equilibrium $v (u)$, the vote total $S_{n}$ satisfies
\begin{equation*}
	\sup_{t\in \mathbb R} |P (S_{n}-E (S_{n}|G)\leq t\sqrt{\var
	(S_{n}|G)}\,|\,G)-\Phi (t)|\leq \kappa n^{-1/2},
\end{equation*}
and consequently, since $1-P (G)=O (n^{-1/2})$, there exists $\kappa
'<\infty$ such that 
\begin{equation*}
	\sup_{t\in \mathbb R} |P \{ S_{n}-E (S_{n}|G)\leq t\sqrt{\var
	(S_{n}|G)}\}-\Phi (t)|\leq \kappa ' n^{-1/2}.
\end{equation*}
Here $\Phi$ denotes the standard normal cumulative distribution
function.
\end{corollary}

\section{Unbalanced Populations: Proofs of  Theorems
  2--3}\label{sec:unbalanced}

\subsection{Concentration of the vote
  total}\label{ssec:concentration-vote-total}

The following lemma shows that if the sampling distribution $F$ has
positive mean $\mu >0$ then when $n$ is large the vote total, in any
Nash equilibrium, concentrates at a value between $\delta$ and
$\delta+ \sqrt{2| \underline{u} |}$. (Recall that $[-\delta,\delta]$
is the support of the function $\psi$.)

\begin{lemma}\label{proposition:concentrate}
If $\mu>0$ then for all large $n$ no Nash equilibrium $v (u)$ has a
discontinuity at a nonnegative value of $u$. Moreover, for any
$\epsilon >0$, if $n $ is sufficiently large then in any Nash
equilibrium the vote total $S_{n}$ of the one-out sample must satisfy
\begin{gather}\label{eq:expectation-bounds}
	\delta -\epsilon \leq ES_{n}\leq \delta +\epsilon
	+\sqrt{2\left|\underline u\right|}] \quad \text{and} \\
\label{eq:vote-concentration}
	P\{|S_{n}-ES_{n}|>\epsilon \}<\epsilon.
\end{gather}
In addition, for any $\epsilon >0$ there exists $\gamma >0$ such that
if $n$ is sufficiently large and $v (u)$ is a Nash equilibrium with no
discontinuities, then 
\begin{equation}\label{eq:exponential-vote-concentration}
	P\{|S_{n}-\delta|>\epsilon \}<e^{-\gamma n}.
\end{equation}
\end{lemma}

The proof of this will rely on a simple estimate for the expected vote
total of the extremists in the out-sample. Denote by $N_{e}$ the
number of extremists among the first $n$ agents, and by $S_{n}'$ their
vote total, formally,
\begin{align}
  \label{eq:n-extremists}
  N_{e}&=\sum_{i=1}^{n} \mathbf{1}_{J_{\max}^{c}}(U_{i})\quad \textrm
  {and}\\
  \label{eq:s-extremists}
  S_{n}'&=\sum_{i=1}^{n}v(U_{i})\mathbf{1}_{J_{\max}^{c}}(U_{i}).  
\end{align}

\begin{lemma}
  \label{lemma:extremist-contrib}
  There exists a constant $C<\infty$ independent of $n$ and of the
  particular Nash equilibrium such that
  \begin{equation}
    \label{eq:small-extremist-contrib}
    E|S_{n}'|\leq Cn^{-1/2} \quad \textrm {and} \quad
    |ES_{n}-E(S_{n}\,|\,G)|\leq Cn^{-1/2}.
  \end{equation}
\end{lemma}

\begin{proof}
  Since not even an extremist would ever buy more than
  $\kappa :=\max(\sqrt{2\left|\underline u\right|},\sqrt{2\overline
    u})$ votes, the extremist total $S_{n}'$ is bounded by
  $\kappa N_{e}$. For  large $n$ the extremist range
  $[\underline{u} ,\overline{u} ]\setminus J_{\max}$ has
  $F-$probability less than
  $2\beta n^{-3/2}(f(\overline{u} )+f(\underline{u} )):= \alpha
  n^{-3/2}$; consequently,
  \begin{equation}\label{eq:NeBound}
    P \xset {N_{e}\geq k} \leq \frac{n^{k}}{k!} (\alpha n^{-3/2})^{k}\leq
    \frac{\alpha^{k}}{k! n^{k/2}} 
     \quad \textrm {for  all} \; k\geq 1.
  \end{equation}
   Thus, if $n$ is sufficiently large,
  \begin{displaymath}
    E |S_{n}'| \leq \alpha\kappa \sum_{k \geq 1}n^{-k/2} \leq 2\alpha \kappa
    n^{-1/2}.
  \end{displaymath}
  This proves the first inequality in
  \eqref{eq:small-extremist-contrib}.

  Suppose now that the $N_{e}$ extremists in the sample were replaced
  by $N_{e}$ independently chosen moderates.  Denote by $S_{n}''$ the
  vote total of these replacements. Clearly, $|S_{n}''|\leq \kappa
  N_{e}$, since no agent ever buys more than $\kappa$ votes, and so by
  the same calculation as above,
  \begin{displaymath}
    E|S_{n}''|\leq 2\alpha\kappa n^{-1/2}.
  \end{displaymath}
  By construction, the sample obtained by replacing the extremists by
  moderates is a sample of size $n$ chosen from $J_{\max}$, and so
  $E(S_{n}-S_{n}'+S_{n}'')=E(S_{n}\,|\, G)$. Consequently, the second
  inequality in \eqref{eq:small-extremist-contrib} follows from the
  bounds on $E|S_{n}'|$ and $E|S_{n}''|$ obtained above.
\end{proof}

\begin{proof}[Proof of Lemma~\ref{proposition:concentrate}]
  By Lemma~\ref{lemma:extremist-contrib}, the difference
  $ES_{n}-E(S_{n}\,|\,G)$ is vanishingly small, and so the expectation
  $E S_{n}$ can be replaced by $E(S_{n}\,|\,G)$ in the inequalities
  \eqref{eq:expectation-bounds}, \eqref{eq:vote-concentration}, and
  \eqref{eq:exponential-vote-concentration}.  By
  Proposition~\ref{proposition:proportionality}, Nash equilibria
  $v (u)$ obey the approximate proportionality
  rule~\eqref{eq:proportionality} in the moderate range, and by
  inequality~\eqref{eq:probability-no-extremists}, the difference
  between $\mu$ and $E(U_{i}\,|\, U_{i}\in J_{\max})$ is negligible; consequently,
  for any $\epsilon >0$, if $n$ is sufficiently large then
  \begin{equation}
    \label{eq:sandwich}
    nE\psi (S_{n}) (\mu-\epsilon) (1-\epsilon)\leq E(S_{n}\,|\,G)\leq
    nE\psi (S_{n}) (\mu+\epsilon) (1+\epsilon). 
  \end{equation}
  By Proposition~\ref{proposition:proportionality}, $E\psi(S_{n})>0$,
  and by hypothesis, $\mu >0$, so it follows that the lower bound in
  \eqref{eq:sandwich} is positive (provided $\epsilon>0$ is
  sufficiently small).

  Suppose now that for some $\epsilon'>0$ there are Nash equilibria
  for arbitrarily large $n$ such that
  $E(S_{n}\,|\,G)<\delta-2\epsilon'$. Since the constant $\epsilon >0$
  in \eqref{eq:sandwich} can be chosen arbitrarily small relative to
  $\epsilon '$, the lower inequality in \eqref{eq:sandwich} implies
  that $n E\psi (S_{n})\mu\leq \delta$; thus, in particular, the
  quantity $nE\psi(S_{n})$ remains bounded as $n
  \rightarrow\infty$. Consequently, the Hoeffding bound
  \eqref{eq:hoeffding-conditional} is exponentially decaying in
  $n$: 
  \begin{displaymath}
    P(|S_{n}-E(S_{n}\,|\,G)| \geq \epsilon) \leq2\exp \xset
    {-2n\epsilon^{2}/(4 [nE\psi(S_{n})]^{2}(\overline{u}
      -\underline{u} )^{2})} \leq 2e^{-\gamma n}  
  \end{displaymath}
  for some constant $\gamma=\gamma(\epsilon')>0$. Since
  $0\leq E(S_{n}\,|\,G)<\delta -2\epsilon' $, it follows that 
  \begin{displaymath}
    P (-\epsilon' < S_{n}<\delta - \epsilon' \,|\, G) \geq 1-2 e^{-n\gamma}.
  \end{displaymath}
  Now $P(G)\rightarrow 1$, by inequality 
  \eqref{eq:probability-no-extremists}, so for large $n$ we have
  \begin{displaymath}
    E\psi (S_{n})\geq \frac{1}{2} \min_{v\in [-\epsilon', \delta	-\epsilon '] }\psi (v) .
  \end{displaymath}
  But since $\psi$ is bounded away from $0$ on the
  interval$[-\epsilon',\delta-\epsilon']$,  this contradicts the fact
  that $nE\psi (S_{n})$ remains bounded as $n \rightarrow\infty$. This
  proves the lower bound in \eqref{eq:expectation-bounds}.

  Next, suppose that the upper bound in \eqref{eq:expectation-bounds}
  were not true, that is, suppose that for some $\epsilon'>0$ and
  indefinitely large $n$ there were Nash equilibria for which
  \begin{displaymath}
    E(S_{n}\,|\,G)>\delta +\sqrt{2\left|\underline u\right|}+\epsilon' .
  \end{displaymath}
  Then  the event $S_{n}\leq \delta+\sqrt{2|\underline{u} |}$ would be
  a large deviation event in \eqref{eq:hoeffding-conditional}: in
  particular, if $\alpha:=\epsilon'/(\delta+\sqrt{2| \underline{u}
    |}+\epsilon')$ then for any $\epsilon>0$, if $n$ is sufficiently large,
  \begin{displaymath}
    |S_{n}-E(S_{n}\,|\,G)|\geq \alpha E(S_{n}\,|\,G) \geq \alpha
    E\psi(S_{n})(\mu-\epsilon) (1-\epsilon),  
  \end{displaymath}
  by inequality \eqref{eq:sandwich}. Consequently, by
  Corollary~\ref{corollary:hoeffding}, there exists
  $\gamma =\gamma (\epsilon ', \epsilon)>0$ such that
  \begin{displaymath}
    P (S_{n}\leq \delta +\sqrt{2\left|\underline	u\right|}\,|\,G)\leq e^{-\gamma n}.
  \end{displaymath}
  But for all $v\geq -\sqrt{2\left|\underline u\right|}$,
  \begin{align*}
    E\psi (v+S_{n})&\leq \xnorm{\psi }_{\infty} P(S_{n}\leq \delta +
                   \sqrt{2|\underline{u} |}\,|\, G)+ \xnorm{\psi }_{\infty}
                   P(G^{c})\\
  &\leq e^{-\gamma n}\xnorm{\psi }_{\infty} + P
	 (G^{c})\xnorm{\psi }_{\infty}=O(n^{-1/2}).
  \end{align*}
  Hence, by the necessary condition \eqref{eq:weierstrass}, the
  function $v (u)$ must be vanishingly small (of order no greater than
  $O (n^{-1/2})$) for all $u\in [\underline u ,\overline u]$, and so
  by Proposition ~\ref{proposition:discontinuities-necessary} can have
  no discontinuities. But then the proportionality rule
  \eqref{eq:proportionality} would hold for \emph{all}
  $u\in \left[\underline u,\overline u \right]$, and so the Hoeffding
  bound \eqref{eq:hoeffding-no-discontinuities} for Nash equilibria
  with no discontinuities would give
  \begin{displaymath}
    P (S_{n}\leq \delta +\sqrt{2\left|\underline u\right|} )\leq
    e^{-\gamma n} \quad	\Longrightarrow \quad	 E\psi (S_{n})\leq
    e^{-\gamma n}\xnorm{\psi }_{\infty}, 
  \end{displaymath}
  which, in view of the approximate proportionality rule~\eqref{eq:proportionality},
  contradicts the \emph{a priori} lower bound on Nash equilibria
  provided by Lemma~\ref{lemma:aPrioriLB}. This proves that for every
  $\epsilon' >0$, if $n$ is sufficiently large then for every Nash
  equilibrium,
  \begin{displaymath}
    E(S_{n}\,|\, G)\leq \delta +\sqrt{2\left|\underline u\right|}+\epsilon ',
  \end{displaymath}
  thus establishing the upper bound in
  \eqref{eq:expectation-bounds}. A similar argument shows that for any
  $\epsilon>0$, if $n$ is sufficiently large then for any
  Nash equilibrium   $v(u)$ with no discontinuities,
  \begin{displaymath}
    \delta-\epsilon \leq E(S_{n}\,|\,G) \leq\delta+\epsilon.
  \end{displaymath}

  Because we have now proved that $E(S_{n}\,|\,G)$ is bounded away
  from $0$ and $\infty$, it follows as before, by the inequalities
  \eqref{eq:sandwich}, that $n E\psi (S_{n})$ is bounded away from $0$
  and $\infty$. Therefore, by the Hoeffding bound
  \eqref{eq:hoeffding-conditional}, for any $\epsilon>0$,
\begin{displaymath}
  P (|S_{n}-E(S_{n}\,|\,G)|>\epsilon \,|\,G)\leq 2\exp \xset
  {-2n\epsilon^{2}/(4 [nE\psi(S_{n})]^{2}(\overline{u} -
    \underline{u} )^{2})} \leq e^{-\gamma n}
\end{displaymath}
for some $\gamma=\gamma(\epsilon)>0$. This proves
\eqref{eq:exponential-vote-concentration}. Furthermore, since
$|ES_{n}-E(S_{n}\,|\,G)|=O(n^{-1/2})$, by
\eqref{eq:small-extremist-contrib}, it follows that
\begin{displaymath}
  P \xset {|S_{n}-ES_{n}|>\epsilon}\leq P
  (|S_{n}-E(S_{n}\,|\,G)|>2\epsilon \,|\,G)+P(G^{c}) =O(n^{-1/2}),
\end{displaymath}
proving \eqref{eq:vote-concentration}.
\end{proof}

\subsection{Identification of the concentration
point}\label{ssec:identification}

\begin{lemma}\label{lemma:identification}
Assume that $\mu >0$, and let $(\xi ,w)$ be the unique solution
of the Optimization Problem (cf. section~\ref{unbalancedsample}), if
one exists, or let $\xi =\delta$ if not. Then for any $\epsilon
>0$, if $n$ is sufficiently large then in every Nash equilibrium,
\begin{equation}\label{eq:xi-claim}
	\bigg\lvert\frac{n}{2} E\psi (S_{n})-\xi
	\mu^{-1}\bigg\rvert<\epsilon . 
\end{equation}
\end{lemma}

\begin{proof}
  The lemma is equivalent to the assertion that
  $|ES_{n}-\xi |\rightarrow 0$, by the proportionality
  rule~\eqref{eq:proportionality} and the estimate
  \eqref{eq:probability-no-extremists}. We will prove this assertion
  in two steps, by first showing that for any $\epsilon>0$ if $n$ is
  sufficiently large the expectation $ES_{n}$ cannot be smaller than
  $\xi -2\epsilon$, and then that it cannot be larger than
  $\xi +2\epsilon$.

  If $\xi =\delta$ then Lemma~\ref{proposition:concentrate} implies
  that $ES_{n}<\xi -2\epsilon$ is impossible for large $n$, so to prove
  that $ES_{n}\geq \xi -2\epsilon$ for large $n$ it suffices to
  consider the case where $\xi >\delta$. In this case $(\xi,w)$ is the
  unique solution to the Optimization Problem \eqref{eq:xi-w}, so for
  sufficiently small $\epsilon >0$ there is a  $\varrho >0$ such that
  \begin{displaymath}
    (1 -\Psi (w))|u|> (\xi -\epsilon -w)^{2} \quad \textrm{for all} \;
    u\in [\underline{u} ,\underline{u} +\varrho].  
  \end{displaymath}
   By
  Lemma~\ref{proposition:concentrate}, for any $\epsilon'>0$, for all
  sufficiently large $n$,
  \begin{displaymath}
    P \xset {|S_{n}-ES_{n}|\geq \epsilon}<\epsilon';
  \end{displaymath}
  hence,   if it were the case that $ES_{n}<\xi - 2\epsilon$ then a voter with value
  $u\in [\underline{u} ,\underline{u} +\varrho]$ could, with
  probability in excess of $1-\epsilon'$, improve her expected utility payoff from
  $\approx u$ to $\Psi (w)u$, at a cost less than
  $(\xi -\epsilon -w)^{2}$, and so all such voters would defect from
  the equilibrium strategy. Since $\epsilon'>0$ can be chosen
  arbitrarily small relative to $\epsilon$, this is a contradiction, because the
  approximate proportionality rule
  (Proposition~\ref{proposition:proportionality}) hold for all
  $u \in J_{\max}$, and $J_{\max} $ overlaps with the interval
  $[\underline{u} , u+\varrho]$ when $n$ is large, by
  Proposition~\ref{proposition:max-cont-interval}. This shows
  that for all large $n$, in any equilibrium,
  $ES_{n}\geq \xi -2\epsilon$.

  Now suppose that $ES_{n}>\xi +2\epsilon$. Then by
  Lemma~\ref{proposition:concentrate}, the one-out vote total $S_{n}$
  would exceed $\xi +\epsilon$ with probability near $1$, for large
  $n$. But by \eqref{eq:xi-w}, for all $w\leq \delta$,
  \begin{displaymath}
    (\xi +\epsilon -w)^{2}\geq \epsilon^{2} + (\xi-w)^{2} \geq
    \epsilon^{2} +(1-\Psi	(w))|\underline u |, 
  \end{displaymath}
  so for a voter with value $u<0$ the cost of buying enough votes to
  move the vote total from above $\xi+\epsilon$ to any value
  $w \in [-\delta ,\delta]$ would exceed the increase in expected
  utility. Hence, for any $\alpha>0$, if $n$ is large then it would be
  suboptimal for a voter with negative value $u$ to buy more than
  $\alpha$ votes. Since this is true, in particular, for
  $\alpha=\Delta/2$, where $\Delta$ is the discontinuity threshold, it
  follows by Proposition~\ref{proposition:discontinuities-necessary}
  that $v(u)$ would have no discontinuity at a value
  $u<0$. Lemma~\ref{proposition:concentrate} guarantees that there are
  no discontinuities at values $u \geq 0$, so it follows that $v(u)$
  is continuous on the entire interval
  $[\underline{u} ,\overline{u} ]$.  But then, because
  $ES_{n}\geq \xi+2\epsilon \geq \delta+2\epsilon$,
  assertion~\eqref{eq:exponential-vote-concentration} implies that
  \begin{displaymath}
    P\{S_{n}\leq\delta \}<e^{-\varrho n}
  \end{displaymath}
  which, since $\psi$ has support $[-\delta ,\delta]$,   implies that
  \begin{displaymath}
    E\psi (S_{n})\leq \xnorm{\psi }_{\infty}e^{-\varrho n}.
  \end{displaymath}
  This is impossible, because the proportionality
  rule~\eqref{eq:proportionality} would then imply that for some
  constant $C<\infty$ not depending on $n$ or the particular
  equilibrium, $\xnorm{v}_{\infty}\leq Ce^{-\varrho n}$, contradicting
  Lemma~\ref{lemma:aPrioriLB}.
\end{proof}

\subsection{Proof of Theorem 2}\label{ssec:proof-efficiency} 

\begin{proof}
[Proof of assertion \eqref{eq:asym-eff}] The asymptotic efficiency of
quadratic voting in the unbalanced case $\mu >0$ is a direct and easy
consequence of Lemma~\ref{proposition:concentrate}. This implies that
for any $\epsilon >0$, the probability that the vote total
$S_{N}=S_{n}+v ( U_{n+1})$ will fall below $\delta -2\epsilon$ is less
than $\epsilon $ for all large $n$, and so by the continuity of the
payoff function $\Psi$, for any $\epsilon >0$
\[
	P \{\Psi (S_{N})\leq 1-\epsilon \}<\epsilon 
\]
for all sufficiently large $N$ and all Nash equilibria. Moreover, the
law of large numbers guarantees that for large $N$,
\[
	P \left\{\lvert N^{-1}U_{+}-\mu \rvert \geq
	\epsilon \right\}<\epsilon  \quad \text{where} \quad
	U_{+}:=\sum_{i=1}^{N}U_{i}  .
\]
Because the random variables $U_{+}/N$ and $\Psi (S_{N})$ are bounded, it
therefore follows that for any $\epsilon >0$, if $N$ is sufficiently
large then in any equilibrium
\[
	\bigg\lvert\frac{E[U\Psi (S_{N})]}{2E|U|} -1 \bigg\rvert <\epsilon .
\]
\end{proof}

\begin{proof}
[Proof of assertions
\eqref{eq:approx-proportionality}--\eqref{eq:vote-total-concentration}]
The second assertion \eqref{eq:vote-total-concentration}  follows
immediately from the first, by the law of large numbers for the
sequence $U_{1},U_{2},\dots$, and the assertion
\eqref{eq:vote-total-concentration}  follows directly from the
approximate proportionality rule \eqref{eq:proportionality} and
Lemma~\ref{lemma:identification}.

\end{proof}

\subsection{Proof of  Theorem
3}\label{ssec:unbalanced-B}

Assume now that $\mu >0$ and that the Optimization Problem
\eqref{eq:xi-w} has a unique solution $(\xi ,w)$ such that
$\xi >\delta$. By Lemmas~\ref{lemma:identification} and
\ref{proposition:concentrate} (in particular,
relation~\eqref{eq:vote-concentration}), together with the approximate
proportionality rule~\eqref{eq:proportionality}, the sum $S_{n}$ of
the one-out sample must concentrate near $\xi$.  But by assertion
\eqref{eq:exponential-vote-concentration} of
Lemma~\ref{proposition:concentrate},  for \emph{continuous} Nash
equilibria the sum $S_{n}$ must concentrate near $\delta$. Therefore,
when the Optimization Problem has a unique solution $(\xi,w)$ with
$\xi>\delta$, all Nash equilibria must have discontinuities, at least
when $n$ is large. Lemma~\ref{proposition:concentrate} guarantees that
there are no discontinuities at nonnegative values $u$, so by
Proposition~\ref{proposition:max-cont-interval}, the discontinuities
must  occur in an interval $[\underline{u} ,\underline{u} +\beta
n^{-3/2}]$ of length $O(n^{-3/2})$.

 Let $v (u)$ be a Nash equilibrium, and let $u_{*}$ be the rightmost
point of discontinuity of $v$.  By
Lemma~\ref{proposition:discontinuities-necessary} and the monotonicity
of Nash equilibria,  $v (u)<-\Delta$ for every $u<
u_{*}$. Obviously, the expected payoff for an agent with value $u$
must exceed the expected payoff under the alternative strategy of
buying no votes. The latter expectation is approximately $\underline
u$, because $S_{n}$ is highly concentrated near $ES_{n}>\xi
-\epsilon$ and so $E\Psi (S_{n})\approx 1$. On the other hand, the
expected payoff at $u<u_{*}$ for an agent playing the Nash strategy
is approximately
\[
	\Psi (\xi  -v (u)) \left(\underline u\right)-v (u)^{2}.
\]
In order that this an improvement over the alternative strategy of
buying no votes, for which the expected payoff is about
$-|\underline{u} |$, it must be the case that
\begin{displaymath}
  |v (u)|\approx \xi -w,
\end{displaymath}
because by hypothesis, $(\xi ,w)$ is
the unique pair such that relations \eqref{eq:xi-w} hold.
This proves assertion (iv) of
Theorem~\ref{munot0characterization}.

It remains to prove assertion (v) of
Theorem~\ref{munot0characterization}, that the rightmost discontinuity
occurs at  $u_{*}\approx \underline{u} +\zeta n^{-2}$.  To accomplish
this, we will first argue that the major contribution to the
expectation $E\psi(S_{n})$ comes from the event $N_{e}=1$ that
there is precisely $1$ extremist in the out-sample. 
Recall first that $E\psi(S_{n})\sim 2\xi /(n\mu)$, by
\eqref{eq:xi-claim}; thus the expectation is of order $O(n^{-1})$. Now
on the event $N_{e}=1$, the vote total of the $n-1$ moderates in the
out-sample will concentrate near $\xi$, by
Lemmas~\ref{lemma:identification} and \ref{proposition:concentrate},
and the single extremist will buy approximately $-\xi+w$ votes, by the
argument above, so conditional on the event $N_{e}=1$ the vote total
$S_{n}$ will be close to $w$, with (conditional) probability
approaching $1$. Consequently, as $n \rightarrow\infty$
\begin{equation}\label{eq:EonN=1}
  E\psi(S_{n})\mathbf{1}_{\xset {N_{e}=1}} \sim \psi (w) P \xset
  {N_{e}=1} \sim \psi (w) n f(\underline{u} )(u_{*}-\underline{u} ).
\end{equation}
It follows that $u_{*}-\underline{u} $ cannot be larger than
$O(n^{-2})$.

Next, consider the event $N_{e}\geq 2$. This event has probability
\begin{align*}
  P \xset {N_{e}\geq 2} &= \sum_{k=2}^{\infty}P \xset {N_{e}=k}\\
  &\leq \sum_{k=2}^{\infty} \frac{n^{k}}{k!} (2f(\underline{u}
  )(u_{*}-\underline{u} ))^{k} \\
  &= O((n (u_{*}-\underline{u} ))^{2} )= O(n^{-3}),
\end{align*}
since $u_{*}-\underline{u} =O(n^{-2})$.  Therefore,
\begin{displaymath}
  E\psi(S_{n})\mathbf{1}_{\xset {N_{e}\geq 2}}\leq
  \xnorm{\psi}_{\infty}P \xset {N_{e}\geq 2} = O(n^{-3})
\end{displaymath}
is of smaller order of magnitude than $E\psi(S_{n})$, which is of size
$O(n^{-1})$.

Lastly, consider the event $G=\xset {N_{e}=0}$. This event has
probability converging to $1$, by inequality
\eqref{eq:probability-no-extremists}, and $|ES_{n}-E(S_{n}\,|\,G)|
\rightarrow 0$, by Lemma~\ref{lemma:extremist-contrib}. Since $ES_{n}
\rightarrow \xi$, by Lemma\ref{lemma:identification}, it follows that
$E(S_{n}\,|\,G)\rightarrow \xi>\delta$. But by
Corollary~\ref{corollary:hoeffding}, the conditional probability that
$S_{n}$ differs from $E(S_{n}\,|\,G)$ by more than $\epsilon$ decays
exponentially with $n$, so for some $\gamma>0$,
\begin{displaymath}
  P(S_{n}< \delta \,|\,G)\leq e^{-\gamma n}
\end{displaymath}
provided $n$ is sufficiently large. This implies that
\begin{displaymath}
  E\psi(S_{n})\mathbf{1}_{\xset {N_{e}=0}} =O(e^{-\gamma n}),
\end{displaymath}
which is negligible compared to $E\psi(S_{n})$, since this is of
magnitude $O(n^{-1})$.

We have now shown that $E\psi(S_{n})\sim
E\psi(S_{n})\mathbf{1}_{\xset {N_{e}=1}}$ as $n \rightarrow
\infty$. Since $E\psi(S_{n})\sim 2\xi/(n\mu)$, by
Lemma~\ref{lemma:identification},  it now follows by relation~\eqref{eq:EonN=1} 
that
\begin{displaymath}
  \psi(w)nf(\underline{u} )(u_{*}-u)\sim \frac{2\xi}{n\mu}.
\end{displaymath}
Thus,
\begin{displaymath}
  n^{2}(u_{*}-\underline{u} ) \longrightarrow \frac{2\xi}{n\mu
    \psi(w)f(\underline{u} )}:= \zeta .
\end{displaymath}

\qed

\section{Balanced Populations:  Proof of Theorem
\ref{mu0characterization}} 
\label{sec:balanced}

\subsection{Continuity of Nash equilibria}\label{ssec:balanced-continuity}

\begin{proposition}\label{theorem:no-extremists}
If $\mu=0$, then for all sufficiently large values of $n $ no Nash
equilibrium $v (u)$ has a discontinuity in $\left[\underline
u,\overline u \right]$. Moreover, for any $\epsilon >0$, if $n$ is
sufficiently large every Nash equilibrium $v (u)$ satisfies
\begin{equation}\label{eq:mean-zero-norm-bound}
	\xnorm{v}_{\infty}\leq \epsilon .
\end{equation}
\end{proposition}

\begin{proof}
The size of any discontinuity is bounded below by a positive constant
$\Delta$, by Lemma~\ref{proposition:discontinuities-necessary}, so it
suffices to prove the assertion \eqref{eq:mean-zero-norm-bound}.  Fix
$\epsilon >0$, and suppose that in some Nash equilibrium there is a
value $u\in [\underline u ,\overline u]$ (necessarily in the
extremist range $[\underline u ,\overline u]\setminus J_{\max}$, by
Proposition~\ref{proposition:max-cont-interval}) such that $|v
(u)|\geq \epsilon$; and hence, 
by the necessary
condition \eqref{eq:weierstrass}, 
\begin{equation}\label{eq:vote-concentration-at-v}
	E\psi (v (u)+S_{n})|u|\geq 2\epsilon.
\end{equation}
On the other hand, 
by Proposition~\ref{proposition:anti-concentration}
there exists $\gamma =\gamma (\epsilon)>0$ such that if $n$ is
sufficiently large then any Nash equilibrium $v (u)$ satisfying
$\xnorm{v}_{\infty}>\epsilon$ must also satisfy $|v (u)|\leq \gamma
/\sqrt{n }$ for all $u$ not within distance $\epsilon$ of one of the
endpoints $\underline u ,\overline u$. Hence, the approximate
proportionality relation \eqref{eq:proportionality} implies
\begin{equation}\label{eq:E-psi-small}
	E\psi (S_{n})\leq \frac{C}{\sqrt{n}}
\end{equation}
for a suitable $C=C (\gamma)<\infty$ independent of $n$. It then
follows from the approximate proportionality rule
\eqref{eq:proportionality} that the inequality $|v(u)|\leq \gamma
/\sqrt{n}$ (possibly with a different constant $\gamma$) holds for all
$u\in J_{\max}$.

We will show that if $\mu =0$ then inequality \eqref{eq:E-psi-small}
is impossible for large $n$; this will imply that the hypothesis
$|v(u)|\geq \epsilon$ for some $u\in[\underline{u} ,\overline{u} ]$ is
untenable. The strategy will be to show that the inequality
$|v(u)|\leq \gamma /\sqrt{n}$, which by
\eqref{eq:vote-concentration-at-v} must hold for all $u\in J_{\max}$,
forces concentration of the distribution of $S_{n}$ in the interval
$[-\delta/2,\delta/2]$.  This in turn keeps $E\psi(S_{n})$ bounded
away from $0$, contradicting \eqref{eq:E-psi-small}.

The necessary condition~\eqref{eq:weierstrass}  and Taylor's theorem
imply that for any $u \in[\underline{u} ,\overline{u} ]$,
\begin{displaymath}
	2v (u) =E\psi (S_{n})u+E\psi '(S_{n}) v (u)u 
	 +\frac{1}{2}
	E\psi '' (S_{n})v (u)^{2}u +\frac{1}{6}E\psi ''' (S_{n}+v(\tilde{u}))v (u)^{3}u
\end{displaymath}
for some $v(\tilde{u})$ intermediate between $0$ and $u$.  By
\eqref{eq:E-psi-small}, the expectation $E\psi(S_{n})$ is of order no
larger than $O(n^{-1/2})$; hence, by Lemma~\ref{lemma:psi-psi-prime},
there exists $C' <\infty$ such that
\begin{displaymath}
  \max (E\psi(S_{n}), E|\psi'(S_{n})| , E|\psi''(S_{n})|) \leq C'/\sqrt{n}.
\end{displaymath}
Since $|v(u)|\leq\gamma/\sqrt{n}$ for all $u\in J_{\max}$, and since
$\psi'''$ is bounded, it follows
that the Taylor expansion can be rewritten as
\begin{displaymath}
  2v(u)= E\psi(S_{n})u +\frac{1}{2} E\psi'(S_{n})E\psi(S_{n})u^{2} + R_{n}(u)
  \quad \textrm {for all} \; u\in J_{\max},
\end{displaymath}
where $|R_{n})(u)|\leq C'n^{-3/2}$ for a constant $C'<\infty$ independent of
$n$ and $u$. Now the extremist region $[\underline{u} ,\overline{u}
]\setminus J_{\max}$ has $F-$probability of order $n^{-3/2}$, by
Proposition~\ref{proposition:max-cont-interval}, and so the event
$G^{c}$ that the sample has at least one extremist has probability of
order no larger than $O(n^{-1/2})$ ; consequently, since
$\mu=EU_{i}=0$, the
second form of the Taylor expansion implies that
\begin{gather}
  \label{eq:cond-mean}
  |2E(S_{n}\,|\,G)- \frac{n}{2}
  E\psi'(S_{n})E\psi(S_{n})\sigma^{2}|=O(n^{-1/2}) \quad \textrm
  {and}\\
  \label{eq:cond-var}
  |4 \var (S_{n}\,|\,G)-n(E\psi(S_{n}))^{2}\sigma^{2}|=O(n^{-1/2})  
\end{gather}

Inequality \eqref{eq:cond-var} implies that $\var (S_{n}\,|\,G)$
remains bounded, by \eqref{eq:E-psi-small}. We will now argue that
$\var (S_{n}\,|\,G)$ must also remain bounded away from $0$. Suppose
not; then for indefinitely large sample sizes $n$ there would be Nash
equilibria for which $\var (S_{n}\,|\,G)$, and hence also
$n(E\psi(S_{n}))^{2}$, approach $0$.  But then, by
Lemma~\ref{lemma:psi-psi-prime} and relation \eqref{eq:cond-mean}, it
would also be the case that $E(S_{n}\,|\,G)$ approaches $0$, and so by
Chebyshev's inequality, for every $\alpha>0$
\begin{displaymath}
  P(|S_{n}|\geq \alpha \,|\, G) \longrightarrow 0.
\end{displaymath}
Since $P(G)\rightarrow 1$, this would force
$E\psi(S_{n})\rightarrow \psi(0)>0$, contradicting
\eqref{eq:E-psi-small}.

Now recall that  the Berry-Essen
theorem  (Corollary~\ref{corollary:berry-esseen}) implies that for
some constant $\kappa<\infty$ independent of $n$,
\[
	|P \{S_{n}-E (S_{n}|G)\leq t\sqrt{\var (S_{n}|G)} \}-\Phi
	(t)|\leq \kappa n^{-1/2}  \quad \text{for all} \; t\in \zz{R}.
\]
 Since the standard normal distribution gives positive probability to
 any nonempty open interval, and since
 $\var(S_{n}\,|\,G)$ and $E(S_{n}\,|\,G)$ remain
 bounded as $n \rightarrow \infty$, with $\var (S_{n}\,|\,G)$ also
 bounded away from $0$, it follows that
 \begin{displaymath}
   \liminf _{n \rightarrow \infty}P \xset {S_{n}\in [-\delta/2 ,\delta/2]}>0.
 \end{displaymath}
 Since $\psi$ is bounded away from $0$ on the interval
 $[-\delta/2,\delta/2]$, this implies that
 \begin{displaymath}
   \liminf _{n \rightarrow \infty}E\psi(S_{n})>0,
 \end{displaymath}
 contradicting \eqref{eq:E-psi-small}.

\end{proof}

\begin{corollary}
  \label{corollary:1st2ndEsts}
  If $\mu=0$ then for any $\epsilon>0$, if $n$ is sufficiently large
  then for any Nash equilibrium $v(u)$ the first and second
  derivatives satisfy
  \begin{equation}
    \label{eq:1st2ndEst}
    1-\epsilon \leq \frac{2v'(u)}{E\psi(S_{n})}\leq 1+\epsilon \quad
    \textrm {and} \quad
    |v''(u)|\leq \epsilon E\psi(S_{n}).
  \end{equation}
\end{corollary}

\begin{proof}
  This follows from the approximate proportionality
  rule~\eqref{eq:proportionality} and the formulas
  \eqref{eq:first-derivative}--\eqref{eq:second-derivative} for the
  derivatives $v',v''$. The key is that for every $\epsilon>0$, if
  $n$ is large then every Nash
  equilibrium satisfies $\xnorm{v}_{\infty}<\epsilon$, by
  Proposition~\ref{theorem:no-extremists}; together with
  equation~\eqref{eq:proportionality}, this implies that
  $0<E\psi(S_{n})<\epsilon $. Since $\epsilon>0$ can be taken
  arbitrarily small, it follows from Lemma~\ref{lemma:psi-psi-prime}
  that $E|\psi'(S_{n})|<\epsilon$. 

  Using formula~\eqref{eq:first-derivative} for $v'(u)$ and Taylor's
  theorem (twice), we obtain
  \begin{multline*}
    2v'(u)=E\psi(S_{n}) +E\psi'(S_{n})v(u)+E\psi'(S_{n})v'(u)u \\
    +\frac{1}{2} E
    \psi''(S_{n}+v(\tilde{u}))v(u)^{2}+E\psi''(S_{n}+v(\hat{u})) v'(u)v(u)u
  \end{multline*}
  where $\tilde{u}$ and $\hat{u}$ are intermediate between $0$ and
  $u$. Since $E\psi(S_{n})$, $E|\psi'(S_{n})|$, and
  $\xnorm{v}_{\infty}$ are all smaller than $\epsilon$, where
  $\epsilon>$ can be taken arbitrarily small, and since
  $|v(u)|\leq \frac{1}{2} E\psi(S_{n})\max (|\underline{u}
  |,\overline{u})$, all terms in the above equation for $2v'(u)$ are
  small compared to the term $E\psi(S_{n})$. This proves the first set
  of inequalities in \eqref{eq:1st2ndEst}.

  The inequality for the second derivative in \eqref{eq:1st2ndEst} can
  be proved in similar fashion. Use
  formula~\eqref{eq:second-derivative} together with Taylor's theorem
  to write
  \begin{multline*}
    2v''(u)=E\psi'(S_{n})(2v'(u)+v''(u)u) + E\psi''(S_{n})v'(u)^{2}u
    \\
    +E\psi''(S_{n}+v(\tilde{u}))v(u)(2v'(u)+v''(u)u)+E\psi'''(S_{n}+v(\hat{u}))v(u)v'(u)^{2}u
  \end{multline*}
  where $\tilde{u}$ and $\hat{u}$ are intermediate between $0$ and
  $u$. As above, for large $n$ the quantities $E\psi(S_{n})$,
  $E|\psi'(S_{n})|$, and $\xnorm{v}_{\infty}$ are all smaller than
  $\epsilon$, where $\epsilon>$ can be taken arbitrarily small. By the
  result of the preceding paragraph,
  $2|v'(u)|\leq 2(1+\epsilon)E\psi(S_{n})$ for large
  $n$. Consequently, all terms in the displayed equation for $2v''(u)$
  are small compared to the first, and so
  \begin{displaymath}
    2v''(u)=E\psi'(S_{n})(2v'(u))(1+o(1)).
  \end{displaymath}
  This proves that $|v''(u)|$ is small compared to $E|\psi'(S_{n})|$
  when $n$ is large.
\end{proof}

Because $\xnorm{v}_{\infty}$ is small for any Nash equilibrium $v$, the
distribution of the vote total $S_{n}$ cannot be too highly
concentrated. This in turn implies the proportionality constant
$E\psi (S_{n})$ in \eqref{eq:proportionality} cannot be too small.

\begin{lemma}\label{lemma:big-var}
For any $C<\infty$ there exists  $n_{C}<\infty$ exists such that for all
$n\geq n_{C}$ and every Nash equilibrium,
\begin{equation}\label{eq:big-var}
	nE\psi (S_{n})\geq C.
\end{equation}
\end{lemma}

\begin{proof}
  Since $\mu=0$, the approximate proportionality
  rule~\eqref{eq:proportionality} and the necessary
  condition~\eqref{eq:weierstrass} imply that for any $\epsilon >0$
  and all sufficiently large $n$,
\[
	|ES_{n}|\leq n\epsilon E\psi (S_{n})E|U| .
\]
Thus, by Hoeffding's inequality (Corollary~\ref{corollary:hoeffding}),
if $nE\psi (S_{n})<C$ then
the distribution of $S_{n}$ must be highly concentrated in a
neighborhood of $0$.  But if this were so we would have, for all large
$n$,
\[
	E\psi (S_{n})\approx \psi (0)>0,
\]
which is a contradiction.
\end{proof}

\subsection{Edgeworth expansions}\label{ssec:edgeworth}

For the analysis of the case $\mu=0$ refined estimates of the
errors in the approximate proportionality rule
\eqref{eq:proportionality} will be necessary. We derive these from the
Edgeworth expansion for the density of a sum of independent,
identically distributed random variables (cf. \cite{feller}, Ch. XVI,
sec. 2, Th. 2).  The relevant summands here are the random variables
$v (U_{i})$, and because the function $v (u)$ depends on the
particular Nash equilibrium (and hence also on $n$), we must employ a
version of the Edgeworth expansion in which the error is precisely
quantified. The following variant of Feller's Theorem 2 (which can be
proved in the same manner as in \cite{feller}) will suffice for our
purposes.

\begin{proposition}\label{theorem:edgeworth}
Let $Y_{1},Y_{2},\dotsc ,Y_{n}$ be independent, identically
distributed random variables with mean $EY_{1}=0$, variance
$EY_{1}^{2}=1$, and finite $2r$th  moment $E|Y_{1}|^{2r}=\mu_{2r}\leq
m_{2r}$. Assume the distribution of $Y_{1}$ has a density $f_{1}
(y)$ whose Fourier transform $\hat{f_{1}}$ satisfies $|\hat{f_{1}}
(\theta)|\leq g (\theta)$, where $g$ is a $C^{2r}$ function such that
$g\in L^{\nu}$ for some $\nu \geq 1$ and such that for every $\epsilon
>0$,
\begin{equation}\label{eq:g-condition}
	\sup_{|\theta |\geq \epsilon}g (\theta)<1.
\end{equation}
Then for some sequence $\epsilon_{n} \rightarrow 0$ depending only
on $m_{2r}$ and on the  function $g$,  the density $f_{n} (y)$
of $\sum_{i=1}^{n}Y_{i}/\sqrt{n}$ satisfies
\begin{equation}\label{eq:edgeworth}
	\bigg|f_{n} (x)-\frac{e^{-x^{2}/2}}{\sqrt{2\pi n}}\left(1+
	\sum_{k=3}^{2r}n^{- (k-2)/2}P_{k} (x)
	\right)
	\bigg| \leq \frac{\epsilon_{n}}{{n}^{r-1}}
\end{equation}
for all $x\in \zz{R}$, where $P_{k} (x)=C_{k}H_{k} (x)$ is a multiple
of the $k$th Hermite polynomial $H_{k} (x)$, and  $C_{k}$
is a continuous function of  the moments $\mu_{3},\mu_{4},\dotsc
,\mu_{k}$ of $Y_{1}$.
\end{proposition}

The following lemma ensures that in any Nash equilibrium the sums
$S_{n}=\sum_{i=1}^{n}v (U_{i})$, after suitable renormalization, meet
the requirements of Proposition~\ref{theorem:edgeworth}.

\begin{lemma}\label{lemma:edgeworth-hypotheses}
There exist
 constants $0<\sigma_{1}<\sigma_{2}<m_{2r}<\infty$ and a
function $g (\theta)$  satisfying the hypotheses of
Proposition~\ref{theorem:edgeworth} (with $r=4$)  such that for all
sufficiently large $n$ and any Nash equilibrium $v (u)$ the following
statement holds.  If $w(u)=2v (u)/E\psi (S_{n})$ 
\begin{compactenum}
\item [(a)] $\sigma_{1}^{2}<\var (w (U_{i}))<\sigma_{2}^{2}$;
\item [(b)] $E|w (U_{i})-Ew (U_{i})|^{2r}\leq m_{2r}$; and
\item [(c)] the random variables $w (U_{i})$ have density $f_{W} (w)$ whose  Fourier
transform is bounded in absolute value by $g$.
\end{compactenum}
\end{lemma}

\begin{proof}
These statements are consequences of the proportionality relations
\eqref{eq:proportionality} and the smoothness of Nash equilibria.  By
Proposition~\ref{theorem:no-extremists}, Nash equilibria are
continuous on $\left[\underline u,\overline u \right]$ and for large $n$ satisfy
$\xnorm{v}_{\infty}<\epsilon$, where $\epsilon >0$ is any small
constant. Consequently, by
Proposition~\ref{proposition:proportionality}, the proportionality
relations \eqref{eq:proportionality} hold on the entire interval
$\left[\underline u,\overline u\right]$. Because $EU_{1}=0$, it
follows that for any $\epsilon >0$, if 
$n$ is sufficiently large then $|Ew (U_{i})|<\epsilon$, and so
assertions (a)--(b) follow routinely from \eqref{eq:proportionality}.

The existence of the density $f_{W} (w)$ follows from the smoothness
of Nash equilibria and smoothness of the sampling density $f$, by
standard change-of-variables rules of calculus.  Any Nash equilibrium
$v(u)$ is continuous on the entire interval
$[\underline{u} ,\overline{u} ]$; hence, by
Proposition~\ref{proposition:max-cont-interval}, $v(u)$ is
$C^{\infty}$ on $(\underline{u} ,\overline{u})$. Thus, if $U$ is
a random variable with density $f (u)$ then the random variable
$W:=2v (U)/E\psi (S_{n})$ has density
\begin{equation}\label{eq:w-density}
	f_{W} (w)=f (u)E\psi (S_{n})/ (2 v' (u)), \quad \text{where} \;
	w=2v (u)/E\psi (S_{n}),
\end{equation}
at every point $w$ such that $v'(u)\not=0$. It follows by
Corollary~\ref{corollary:1st2ndEsts} that for any $\epsilon>0$, if $n$
is sufficiently large then
\begin{displaymath}
  (1-\epsilon)f(u)\leq f_{W}(w)\leq(1+\epsilon)f(u) \quad \text{where}
  \;	w=2v (u)/E\psi (S_{n}).
\end{displaymath}
Furthermore, since the functions $f(u)$, $v(u)$, and $\psi$ are all
$C^{\infty}$ with compact support, equation~\eqref{eq:w-density}
implies that the density $f_{W}(w)$ is continuously differentiable,
with derivative
\begin{displaymath}
  f_{W}' (w)= \frac{f' (u) (E\psi (S_{n}))^{2}}{4v' (u)^{2}}-\frac{f
    (u) (E\psi (S_{n}))^{2}v'' (u)}{4v' (u)^{3}}. 
\end{displaymath}
By Corollary~\ref{corollary:1st2ndEsts} implies that the ratio
$(E\psi(S_{n}))^{2}/4v'(u)^{2}$ is bounded above and below by $1\pm
\epsilon$ when $n$ is large, and also that $|v''(u)|$ is small
compared to $|v'(u)|$, so it follows that for large $n$ the ratio $|f'_{W}(w)|/|f'(u)|$
is uniformly close to $1$. Thus, in particular,
\begin{equation}\label{eq:density-derivative-bound}
	|f_{W}' (w)|\leq \kappa 	
\end{equation}
where $\kappa <\infty$ is a constant that does not depend on either
$n$ or on the choice of Nash equilibrium.

The last step is to prove the existence of a dominating function $g(\theta)$
for the Fourier transform of $f_{W}$. This will rely on the
differentiablility of the density  $f_{W}(w)$ and the
inequality~\eqref{eq:density-derivative-bound}. We will analyze the
Fourier transform in three
regions: (i) for values $|\theta |\leq \gamma $, where $\gamma  >0$
is a small fixed constant; (ii) for values $|\theta |\geq K$, where
$K$ is a large but fixed constant; and (iii) for $\gamma  <|\theta
|<K$. Region (i) is easily dealt with, in view of the bounds (a)--(b)
on the second and third moments and the estimate $|Ew (U)|<\epsilon'$,
as these together with Taylor's theorem imply that for all $|\theta
|<1$,
\[
	|\hat{f}_{W} (\theta)- (1+i\theta Ew (U) -\theta^{2}\var (w
	(U))/2 |\leq m_{3} |\theta |^{3}.
\]
Next consider region (ii), where $|\theta |$ is large. Integration by
parts shows that
\[
	\hat{f}_{W} (\theta)=\int_{w (\underline u )}^{w (\overline u)}
	f_{W} (w)e^{i\theta w}\,dw = -\int_{w (\underline u )}^{w
	(\overline u)} \frac{e^{i\theta w}}{i\theta} f_{W}' (w)\,dw + 
	\frac{e^{i\theta w}}{i\theta} f_{W} (w) \bigg\vert_{w
	\underline u}^{w (\overline u)};	 
\]
because $f_{W} (w)$ is uniformly bounded at $w (\underline u )$ and $w (\overline u)$, by
\eqref{eq:derivative-bounds} and \eqref{eq:w-density}, and because
$|f_{W}' (w)|\leq \kappa$, by \eqref{eq:density-derivative-bound}, it
follows that there is a constant $C<\infty$  such that for all
sufficiently large $n$ and all Nash equilibria,
\[
	|\hat{f}_{W} (\theta)|\leq C/|\theta | \quad \forall \; \theta \not =0.
\]
Thus,  setting $g (\theta)=C/|\theta |$ for all $|\theta |\geq 2C$,
we have a uniform bound for the Fourier transforms $\hat{f}_{W}
(\theta)$ in the region (ii).

Finally, to bound $|\hat{f}_{W} (\theta)|$ in the region (iii) of
intermediate $\theta -$values, we use the proportionality rule once
again in the form $|w (u)-u|<\epsilon$, valid for all $u\in
[\underline u ,\overline u]$. This implies
\begin{align*}
	\hat{f}_{W} (\theta) &=\int_{\underline u}^{\overline u} e^{i\theta w (u)}f (u) \,du\\
		    &=\int_{\underline u}^{\overline u} e^{i\theta u}f (u) \,du
		    +\int_{\underline u}^{\overline u} (e^{i\theta w (u)}-e^{i\theta u})
		    f (u)\,du \\
		    &=\hat{f}_{U} (\theta) + R (\theta)
\end{align*}
where $|R (\theta)|<\epsilon '$ uniformly for $|\theta |\leq C$ and
$\epsilon ' \rightarrow 0$ as $\epsilon \rightarrow 0$. Because
$\hat{f}_{U}$ is the Fourier transform of an absolutely continuous
probability density, its absolute value is bounded away from $1$ on
the complement of $[-\gamma ,\gamma]$, for any $\gamma >0$. Since
$\epsilon >0$ can be made arbitrarily small
(cf. Proposition~\ref{proposition:proportionality}), it follows that
 a continuous, positive  function $g (\theta )$ that is bounded away
from 1 on $|\theta |\in [\gamma ,C]$ exists such that
$|\hat{f}_{W})\theta|\leq g (\theta)$ for all $|\theta |\in [\gamma ,C]$.
The extension of $g$ to the whole real line can now be done by
smoothly interpolating at the boundaries of regions (i), (ii), and (iii).
\end{proof}

\subsection{Proof of  Theorem
1}\label{ssec:proof-balanced}

Because the function $\psi$ is smooth and has compact support,
differentiation under the expectation in the  necessary condition $2v
(u)=E\psi (v (u)+S_{n})u$  is permissible, and so for every $u\in
[-\underline u ,\overline u]$ a $\tilde{v} (u)$ exists
intermediate between $0$ and $v (u)$ such that
\begin{equation}\label{eq:taylor}
	2v (u)=E\psi (S_{n})u+E\psi ' (S_{n}+\tilde{v} (u))v (u)u .
\end{equation}
The proof of  Theorem 1 will hinge on the use
of the Edgeworth expansion (Proposition~\ref{theorem:edgeworth}) to
approximate each of the two expectations in \eqref{eq:taylor}
precisely.

As in Lemma~\ref{lemma:edgeworth-hypotheses}, let $w(u)=2v (u)/E\psi
(S_{n})$.  We have already observed, in the proof of
Lemma~\ref{lemma:edgeworth-hypotheses}, that for any $\epsilon >0$, if
$n$ is sufficiently large then for any Nash equilibrium, $|Ew
(U)|<\epsilon$. It therefore follows from the proportionality rule
that
\begin{equation}\label{eq:moment-proportionality}
	  \bigg|\frac{4\,\var (v (U))}{(E\psi
	  (S_{n}))^{2}\sigma_{U}^{2}}-1\bigg|\leq \epsilon
	  \quad \text{and} \quad
	  \bigg| \frac{E |v (u)-Ev (u)|^{k}}{( E\psi
	  (S_{n}))^{k}E|U|^{k}} \bigg|< \epsilon \; \; \forall \,k\leq 8.
\end{equation}
Moreover, Lemma~\ref{lemma:edgeworth-hypotheses} and
Proposition~\ref{theorem:edgeworth} imply the distribution of $S_{n}$
has a density with an Edgeworth expansion, and so for any
continuous function $\varphi :[-\delta ,\delta]\rightarrow \zz{R}$,
\begin{equation}\label{eq:edgeworth-integral}
	E\varphi (S_{n})=\int_{-\delta}^{\delta} \varphi (x)
	\frac{e^{-y ^{2}/2}}{\sqrt{2\pi n}\sigma_{V}}
	\left(1+ \sum_{k=3}^{m}n^{- (k-2) /2}P_{k} (y)\right) \,dx
	+r_{n} (\varphi )
\end{equation}
where
\begin{align*}
	&\sigma_{V}^{2}:=\var (v (U)),\\
	&y=y (x)= (x-ES_{n})/\sqrt{\var (S_{n})},
\end{align*}
and $P_{k} (y)=C_{k}H_{3} (y)$ is a multiple of the $k$th  Hermite
polynomial. The constants $C_{k}$ depend only on the first $k$ moments
of $w (U)$, and consequently are uniformly bounded by  constants
$C_{k}'$ not depending on $n$ or on the choice of Nash
equilibrium. The error term $r_{n} (\varphi)$ satisfies
\begin{equation}\label{eq:remainder}
	|r_{n} (\varphi  )|\leq
	\frac{\epsilon_{n}}{n^{(m-2)/2}}\int_{-\delta}^{\delta} \frac{|\varphi
	(x)|}{\sqrt{2\pi \var (S_{n})}} \,dx.
\end{equation}
In the special case $\varphi =\psi$, \eqref{eq:edgeworth-integral} and the remainder estimate
\eqref{eq:remainder} (with $m=4$) imply 
\[
	E\psi (S_{n})\leq \frac{1}{\sqrt{2\pi
	n}\sigma_{V}}\int_{-\delta}^{\delta} \psi (x)\,dx  +o (n^{-1}\sigma_{V}^{-1}).
\]
Because $4\, \sigma_{V}^{2}\approx (E\psi (S_{n}))^{2}\sigma_{U}^{2}$
for large $n$, this implies that for a suitable constant $\kappa
<\infty$,
\begin{equation}\label{eq:not-too-big}
	E\psi (S_{n})\leq \frac{\kappa}{\sqrt[4]{n}}.
\end{equation}

\begin{claim}\label{claim:mean-var}
There exist constants $\alpha_{n}\rightarrow \infty$  such that in
every Nash equilibrium,
\begin{gather}\label{eq:meanBound}
	|ES_{n}|\leq \alpha_{n}^{-1}\sqrt{\var (S_{n})} \quad \text{and}\\
\label{eq:varBound}
	\var (S_{n})\geq \alpha_{n}^{2}.
\end{gather}
\end{claim}

\begin{proof}
[Proof of  Theorem \ref{mu0characterization}: Conclusion]
Before we begin the proof of the claim, we indicate how it will imply 
Theorem \ref{mu0characterization}. If \eqref{eq:meanBound} and
\eqref{eq:varBound} hold, then for every $x\in [-\delta ,\delta]$,
\[
	|y (x)|\leq (1+2\delta )/\alpha_{n} \rightarrow 0.
\]
Consequently, the dominant term in the Edgeworth expansion
\eqref{eq:edgeworth-integral} for $\varphi =\psi$ (with $m=4$), is the first,
and so for any $\epsilon >0$, if  $n$ is sufficiently large,
\[
	E\psi (S_{n})=\frac{1}{\sqrt{2\pi
	n}\sigma_{V}}\int_{-\delta}^{\delta} \psi (x) \,dx (1 \pm \epsilon ).
\]
(Here the notation $(1\pm \epsilon)$ means the ratio of the two
sides is bounded above and below by $(1\pm \epsilon)$.)
Thus $4\, \sigma_{V}^{2}\approx (E\psi (S_{n}))^{2}\sigma_{U}^{2}$
implies
\[
	\sqrt{\pi n/2}\sigma_{U} (E\psi (S_{n}))^{2} =
	\int_{-\delta}^{\delta}\psi (x) \,dx (1\pm \epsilon)
	=2\pm 2\epsilon ,
\]
proving the assertion \eqref{eq:e-vote-total}.
\end{proof}

\begin{proof}
[Proof of Claim~\ref{claim:mean-var}]
First we deal with the remainder term $r_{n} (\varphi)$ in the
Edgeworth expansion \eqref{eq:edgeworth-integral}. By
Lemma~\ref{lemma:big-var}, the expectation $E\psi (S_{n})$ is at least
$C/n$ for large $n$, and so by \eqref{eq:moment-proportionality} the
variance of $S_{n}$  must be at least $C'/n$. Consequently, by
\eqref{eq:remainder}, the remainder term $r_{n} (\varphi)$ in
\eqref{eq:edgeworth-integral} satisfies
\[
	|r_{n} (\varphi )|\leq C''  \frac{\epsilon_{n}\xnorm{\varphi}_{1}}{n^{(m-2)/2}\sqrt{\var (S_{n})}}\\
	       \leq C''' \frac{\epsilon_{n}\xnorm{\varphi}_{1}}{n^{(m-3)/2}}.
\]
Suitable choice of $m$ will make this bound small compared to any
desired monomial $n^{-A}$, and so we may ignore the remainder term in
the arguments to follow.

Suppose there were a constant $C<\infty$  such that along some
sequence $n \rightarrow \infty$  Nash equilibria existed satisfying
$\var (S_{n})\leq C$. By \eqref{eq:moment-proportionality}, this would
force $C/n\leq E\psi (S_{n})\leq C'/\sqrt{n}$, which in turn would
imply that
\begin{equation}\label{eq:mean-too-big}
	 C''\var (S_{n})\log n \geq |ES_{n}|^{2}\geq C'''\var (S_{n})\log n ,
\end{equation}
because otherwise the dominant term in the Edgeworth series for $E\psi
(S_{n})$ would be either too large or too small asymptotically (along
the sequence $n \rightarrow \infty$) to match the requirement that
$C/n\leq E\psi (S_{n})\leq C'/\sqrt{n}$.  (Observe that because the
ratio $|ES_{n}|^{2}/\var (S_{n})$ is bounded above by $C''\log n$,
the terms $e^{-y^{2}/2}P_{k} (y)$ in the integral
\eqref{eq:edgeworth-integral} are of size at most $(\log n)^{A}$ for
some $A$ depending on $m$, and so the first term in the Edgeworth
series is dominant.)  We will show that \eqref{eq:mean-too-big} leads
to a contradiction.

Suppose  $ES_{n}>0$ (the case $ES_{n}<0$ is similar). The Taylor
expansion \eqref{eq:taylor} for $v (u)$ and the hypothesis  $EU=0$
implies 
\begin{equation}\label{eq:vMean}
	2Ev (U)=E\psi ' (S_{n}+\tilde{v} (U))v (U)U.
\end{equation}
The Edgeworth expansion
\eqref{eq:edgeworth-integral}  for $E\psi ' (S_{n}+\tilde{v} (u))$
together with the independence of $S_{n}$ and $U$ and the inequalities
\eqref{eq:mean-too-big}, implies that  for any $\epsilon >0$, if $n$
is sufficiently large then
\begin{multline}\label{eq:psiPrimeMean}
	E\psi ' (S_{n}+\tilde{v} (u))\\
	 =\frac{1}{\sqrt{2\pi \var (S_{n})}}\int_{-\delta}^{\delta} \psi '
	(x) \exp \{- (x+\tilde{v} (u)-ES_{n})^{2}/2\var (S_{n})\}
	 \,dx (1\pm \epsilon ).
\end{multline}
Now because $\psi$ and $\psi '$ have support $[-\delta ,\delta]$,
integration by parts yields
\begin{multline}\label{eq:integrateByParts}
	\int_{-\delta}^{\delta} \psi '
	(x) \exp \{- (x+\tilde{v} (u)-ES_{n})^{2}/2\var (S_{n})\}
	 \,dx \\
	  =
	 \int_{-\delta}^{\delta} \psi (x) \exp \{- (x+\tilde{v}
	 (u)-ES_{n})^{2}/2\var (S_{n})\} \frac{x+\tilde{v} (u)-ES_{n}}{\var (S_{n})}
	 \,dx,
\end{multline}
and because $x+\tilde{v} (u)$ is of smaller order of magnitude than
$ES_{n}$, it follows that for large $n$
\begin{equation}\label{eq:IBP-Consequence}
	E\psi ' (S_{n}+\tilde{v} (u))=-\frac{ES_{n}}{\var (S_{n})}
	E\psi (S_{n}) (1\pm \epsilon).
\end{equation}
But it now follows from the Taylor series for $2Ev (U_{i})$ (by
summing over $i$) that
\begin{equation}\label{eq:IBP-Consequence-Contra}
	2ES_{n}=-n\frac{ES_{n}}{\var (S_{n})}E\psi (S_{n}) Ev (U)U (1\pm \epsilon ),
\end{equation}
which is a contradiction, because the right side is negative and the
left side positive. This proves the assertion \eqref{eq:varBound}.

The proof of inequality \eqref{eq:meanBound} is similar. Suppose 
for some $C>0$  Nash equilibria existed along a sequence $n
\rightarrow \infty$ for which $ES_{n}\geq C\sqrt{\var (S_{n})}$. In
view of  \eqref{eq:varBound}, this hypothesis implies in particular that
$ES_{n}\rightarrow \infty$, and also that $|y (x)|\geq C/2$ for all
$x\in [-\delta ,\delta]$. Thus, the Edgeworth approximation
\eqref{eq:psiPrimeMean} remains valid, as does the integration by parts
identity \eqref{eq:integrateByParts}. Because $ES_{n}\rightarrow
\infty$, the terms $x+\tilde{v} (u)$ are of smaller order of magnitude
that $ES_{n}$, and so once again
\eqref{eq:IBP-Consequence} and therefore
\eqref{eq:IBP-Consequence-Contra} follow. Again we have a
contradiction,  because the right side of
\eqref{eq:IBP-Consequence-Contra} is negative while the left side
diverges to $+\infty$.

\end{proof}

\section{Conclusion}\label{conc}
In this paper we have shown the asymptotic efficiency of QV in a canonical non-cooperative independent private values environment.  We hope future research will further clarify the performance of QV in environments where collusion is possible, with aggregate uncertainty and partially common values and under a wider range of distributional assumptions.  Hopefully such analyses will inform and be guided by on-going field experiments with QV.

\newcommand{\E}{\mathbb{E}}
\newcommand{\V}{\mathbb{V}}

\appendix

\section{Strict Monotonicity  of Nash
Equilibria}\label{sec:monotonicity}

\begin{lemma}\label{proposition:no-zeros}
If  $v (u)$ is a Nash
equilibrium then $v (u)\not =0$ for
all $u\not =0$.
\end{lemma}

\begin{proof}
  Since any Nash equilibrium $v(u)$ is a nondecreasing function of $u$
  (by Proposition~\ref{proposition:existence}), if $v (u)=0$ for
  some $u>0$ then $v (u')=0$ for all $u'\in (0,u)$. Because the
  density $f (u)$ of the value distribution $F$ is strictly positive
  on $\left[\underline u,\overline u\right]$, it follows that the
  probability $p$ that every agent in the sample casts vote $V_{i}=0$
  is strictly positive. But then an agent with utility $u$ could
  improve her expectation by buying $\varepsilon >0$ votes, where
  $\varepsilon \ll u\psi (0)p$, because the expected utility gain
  would be at least
\[
	u \Psi (\varepsilon) p\sim  u\psi (0)p\varepsilon
\]
at a cost of $\varepsilon^{2}$. Because by hypothesis $\psi (0)>0$, the
expected utility gain would overwhelm the increased vote cost for
small $\varepsilon >0$.
\end{proof}

\begin{corollary}\label{corollary:strict-monotonicity} Any Nash
equilibrium $v (u)$ is \emph{strictly} increasing on $\left[\underline
u,\overline u\right]$.
\end{corollary}

\begin{proof}
  Proposition  \ref{proposition:existence} implies that $v(u)$ is
  nondecreasing in $u$, so it suffices to show that $v(\cdot )$ takes
  distinct values at distinct arguments $u_{1}\not=u_{2}$. By
  Lemma~\ref{proposition:no-zeros}, 
   $v(u)>0$ for $u>0$ and $v(u)<0$ for $u<0$; hence, by  the
  necessary condition \eqref{eq:weierstrass}, for every $u\not = 0$,
  \begin{displaymath}
    E\psi (S_{n}+v (u))>0 .
  \end{displaymath}

  Let $u_{1}\not = u_{2}$ be two distinct nonzero values. By the
  necessary condition  \eqref{eq:weierstrass},
  \begin{align*}
    E\psi (S_{n}+v (u_{1}))u_{1}&=v(u_{1})\quad \textrm {and}\\
    E\psi (S_{n}+v (u_{2}))u_{2}&=v(u_{2});
  \end{align*}
  since $E\psi (S_{n}+v (u_{1}))\not = 0$, it follows that
  \begin{displaymath}
    E\psi (S_{n}+v (u_{1}))u_{2}\not=v(u_{1}) \quad \Longrightarrow
    \quad v(u_{2})\not=v(u_{1}).
  \end{displaymath}
  Thus, the function $v(u)$ takes distinct values at distinct
  arguments $u$. By Lemma~\ref{proposition:existence}, $v(u)$ is
  nondecreasing in $u$; therefore, $v(u)$ is strictly increasing.  
\end{proof}

\section{Necessary Condition for a Nash Equilibrium}\label{sec:weierstrass}

\begin{proof}[Proof of Proposition \ref{proposition:weierstrass}]
  For an agent with value $u>0$ a best response $v=v(u)$ must maximize
  expected utility minus vote cost \eqref{eq:wealth-maximization},
  and so  for every $\Delta >0$, 
  \begin{align*}
    E\left\{\Psi (S_{n}+v+\Delta)-\Psi (S_{n}+v) \right\}u&\leq
                                                            2\Delta v +\Delta^{2} \quad \text{and}\\
    \notag 	E\left\{\Psi (S_{n}+v-\Delta)-\Psi (S_{n}+v)
    \right\}u&\leq 	-2\Delta v +\Delta^{2}.
  \end{align*}
  Dividing by $\Delta$ and letting $\Delta \rightarrow 0$ yields
  \begin{align*}
    2v&\geq \limsup_{\Delta \rightarrow 0 +} \frac{1}{\Delta}(
       E\Psi (S_{n}+v+\Delta)-E\Psi (S_{n}+v) ) \quad
       \textrm{and}\\
    -2v &\geq \limsup_{\Delta \rightarrow 0 +} \frac{1}{\Delta}(
       E\Psi (S_{n}+v-\Delta)- E\Psi (S_{n}+v)).
  \end{align*}
  Because $\Psi$ is continuously differentiable with compactly
  supported derivative $\psi$, the limsups can be taken under the
  expectations, where they become limits, and so
  \begin{displaymath}
    2v= E\psi(v+S_{n}).
  \end{displaymath}
  A similar argument applies to $u<0.$
\end{proof}

\section{Hoeffding's Inequality}\label{sec:hoeffding}

Hoeffding's inequality \cite{hoeffding} is a substantial sharpening of the Chebyshev
bound for sums of \emph{bounded} independent random variables.
Following are two variants of the inequality  adapted to the
needs of this paper; the second shows that for sums of bounded,
i.i.d. random variables the probabilities of ``large deviations'' are
exponentially decaying in the number of summands.

\begin{hoeffding}
  Let $X_{1},X_{2},\cdots , X_{n}$ be independent random variables
  satisfying $a \leq X_{i}\leq b$, where $a<b$ are two real
  constants. Then for any real $t>0$,
  \begin{align}
    \label{eq:hoeffding-ineq}
    P\xset{|S_{n}- E S_{n}| \geq t}&\leq 2 e^{-2t^{2}/n(b-a)^{2}} \quad
    \Longrightarrow \\
    P\xset{|S_{n}- E S_{n}| \geq nt}&\leq 2 e^{-2nt^{2}/(b-a)^{2}}
  \end{align}
\end{hoeffding}

\section{The Berry-Esseen Theorem}\label{sec:berry-esseen}

The Berry-Esseen theorem provides sharp bounds on the error in the
central limit theorem. For the proof and further discussion, see
\cite{feller}, section XVI. 5.

\begin{berryEsseenTh}
  Let $Y_{1},Y_{2}, \cdots $ be independent, identically distributed
  random variables such that
  \begin{equation}
    \label{eq:berryEsseenTh-hyp}
    EY_{1}=0, \quad EY_{1}^{2}=\sigma^{2}>0, \quad \textrm{and} \;
    E|Y_{1}|^{3}=\gamma <\infty.
  \end{equation}
  Then for all $x\in\mathbb{R}$ and all inegers $n \geq 1$,
  \begin{equation}
    \label{eq:berryEsseenTh}
    \bigg\vert P \xset {\frac{1}{\sigma\sqrt{n}}\sum_{i=1}^{n}Y_{i} \leq
      x}- \Phi (x) 
    \bigg \vert \leq \frac{3\gamma}{\sigma^{3}\sqrt{n}},
  \end{equation}
  where $\Phi$ denotes the standard normal cumulative distribution
  function.
\end{berryEsseenTh}

\begin{proof}[Proof of Proposition~\ref{proposition:llt}]
  
  First, we claim that it suffices to prove the inequality
  \eqref{eq:small-local-prob} for intervals of length $\alpha$. To see
  this, observe that any interval $J$ of length larger than $\alpha$
  is contained in an interval $J'$ whose length $|J'|=n\alpha$ is an
  integer multiple of $\alpha$ satisfying $n \alpha \leq 2
  |J|$. Clearly, the probability that
  $S^{Y}_{n}:=\sum_{i=1}^{n} Y_{i}$ falls in $J$ is no larger than the
  probability that it falls in $J'$, and by the Bonferroni inequality
  this probability is no larger than $n$ times the maximal probability
  over all intervals of length $\alpha$; thus, if the inequality holds
  for intervals of length $\alpha$ then it will hold (with $\epsilon$
  replaced by $2\epsilon$) for all intervals of length $\geq\alpha$ .
  
  Second, we can assume, without loss of generality, that $EY_{1}=0$
  and $\alpha =1$, because this can always be accomplished by
  translation and rescaling.  Thus, we must show that for any
  $\epsilon >0$ and $C<\infty$ there exist $C'=C'(\epsilon,C)$ and
  $n'=n'(\epsilon,C)$ such that if $n \geq n'$ and if the summands
  $Y_{i}$ are i.i.d. and satisfy the moment constraints
  \eqref{eq:var-to-third}, then for every real number $a$,
\begin{displaymath}\label{eq:llt-Obj}
  P \xset {S^{Y}_{n} \in [a,a+1]}\leq\epsilon .
\end{displaymath}

Let $\sigma^{2}=EY_{i}^{2}$ be the variance and $\gamma=E|Y_{i}|^{3}$
the third absolute moment of the summands. By hypothesis,
$\gamma \leq C\sigma^{3}$; consequently,
by the Berry-Esseen theorem,
\begin{displaymath}
  P \xset {S^{Y}_{n}\in [a,a+1]}\leq \Phi (\frac{a+1}{\sigma\sqrt{n}})-\Phi(\frac{a}{\sigma\sqrt{n}})
  + 6C /\sqrt{n}. 
 \end{displaymath}
 If the variance satisfies $\sigma^{2} \geq C' n$ then the interval
 $[a/\sigma \sqrt{n},(a+1)/\sigma \sqrt{n}]$ has length
 $(\sigma \sqrt{n})^{-1}$ bounded above by $1/\sqrt{C'}$; thus, if
 $C'$ is chosen large enough that $1/\sqrt{2 \pi C'}<\epsilon/2$ then
 for every $a \in \mathbb{R}$,
 \begin{displaymath}
    \Phi (\frac{a+1}{\sigma\sqrt{n}})-\Phi(\frac{a}{\sigma\sqrt{n}})
    =\frac{1}{\sqrt{2\pi}}\int_{\frac{a}{\sigma\sqrt{n}}}^{\frac{a+1}{\sigma\sqrt{n}}}e^{-t^{2}/2} 
   \,dt \leq \epsilon /2. 
 \end{displaymath}
Finally, if $n'$ is chosen so large  that $6C/\sqrt{n'}<\epsilon/2$,
then for all $n \geq n'$ we will have
\begin{displaymath}
  P \xset {S^{Y}_{n}\in [a,a+1]}\leq
  \Phi (\frac{a+1}{\sigma\sqrt{n}})-\Phi(\frac{a}{\sigma\sqrt{n}})+ 6C
  /\sqrt{n}\leq\epsilon .
 \end{displaymath}
\end{proof}

\section{Proof of  Proposition~\ref{proposition:min-problem}
}\label{sec:min-problem} 

We shall assume throughout that $\delta <1/\sqrt{2}$, and that the
function $\Psi$ and its derivatives satisfy the model assumptions
(M1)--(M4) of section~\ref{ssec:model}. Thus, $\psi/2$ is an even,
$C^{\infty}$ probability density with support $[-\delta ,\delta]$; it
has positive derivative $\psi '$ on $(-\delta ,0)$ (and hence negative
derivative on $(0,\delta)$); and there is a single point $\iota$ of
inflection in the interval $(-\delta ,0)$ such that $\psi'$ is
strictly increasing in $[-\delta,\iota]$ and strictly decreasing in
$[\iota ,0]$

Define 
\begin{equation}\label{eq:H}
	H (\alpha ,w)= (1-\Psi (w))|\underline u | - (\alpha -w)^{2}.
\end{equation}
Proposition~\ref{proposition:min-problem} asserts that, under the
assumption $\delta <1/\sqrt{2}$, there is a unique value $\xi >\delta$
such that (i) the maximum value of the function $w\mapsto H (\xi ,w)$
for $w\in [-\delta,\delta]$ is $0$, and (ii) this maximum is attained
at at a unique point $w\in (-\delta ,\delta)$. The next lemma
establishes the uniqueness of the value $\xi$; Lemma~\ref{lemma:w}
below will show the  uniqueness of the  maximizer $w$.

\begin{lemma}
  \label{lemma:unique-Opt}
  There is a unique $\xi>\delta$ such that
  $\max _{w\in[-\delta,\delta]}H(\xi,w)=0$. Moreover,
  \begin{compactenum}
  \item [(U1)] if $\alpha>\xi$ then
    $\max_{w\in[-\delta,\delta]}H(\alpha,w)<0$; and 
  \item [(U2)] if $\alpha <\xi$ then $\max_{w\in[-\delta,\delta]}H(\alpha,w)>0$.
  \end{compactenum}
\end{lemma}

\begin{proof}
  For each fixed $w\in[-\delta,\delta]$ the function
  $\alpha\mapsto H(\alpha,w)$ is {strictly} decreasing in the
  interval $\alpha\in [\delta,\infty)$, because its derivative
  $-2(\alpha-w)$ is negative throughout this interval. Hence,  since
  $H$ is jointly continuous in its arguments and since the interval
  $[-\delta,\delta]$ is compact, the function
  \begin{displaymath}
    h(\alpha):= \max_{w\in[-\delta,\delta]}H(\alpha,w)
  \end{displaymath}
  is continuous and strictly decreasing for
  $\alpha\in[\delta,\infty)$. Thus,  to complete the proof
  it suffices to show that there exists $\xi\in [\delta,\infty )$
  such that $h(\xi)=0$.

  Clearly, $\lim_{\alpha
    \rightarrow\infty}h(\alpha)=-\infty$, because to first order the
  maximum value of $H(\alpha,w)$ over $w\in[-\delta,\delta]$ is
  controlled by the quadratic term in \eqref{eq:H}, so by the
  Intermediate Value Theorem of calculus it is enough to show that
  $h(\delta)>0$. This is where the hypothesis that $\delta<1/\sqrt{2}$
  comes in, as it implies that $(\alpha-(-\delta))^{2}=4\delta^{2}<2$
  when $\alpha=\delta$. By the standing model assumptions
  (cf. sec.~\ref{ssec:model}), $|\underline{u} |\geq 1$, so
  \begin{displaymath}
    (1-\Psi (-\delta))|\underline{u} |=2|\underline{u} |\geq 2 \quad \Longrightarrow
    H(\delta,-\delta)> 0.
  \end{displaymath}
\end{proof}

\begin{lemma}
  \label{lemma:w} Under the standing hypotheses (M1)--(M4) of
  section~\ref{ssec:model}, there is a unique point $w_{*}\in (-\delta,\delta)$
  where the function $w\mapsto H(\xi,w)$ attains the value $0$.
\end{lemma}

\begin{proof}
  Local minima and maxima of the function $w\mapsto H(\xi,w)$ must be
  critical points, that is, points $w$ where the first partial derivative
  $\partial H/\partial w$ vanishes. We will argue that there are
  either two or three critical points  in the interval
  $(-\delta,\delta)$, and that one  of these is
  the unique point in $[-\delta,\delta]$ where $w\mapsto H(\xi,w)$
  achieves its maximum value.

  The  first and second partial derivatives of $H(\xi,w)$ with respect
  to $w$ are
  \begin{align*}
    \frac{\partial H}{\partial w}= -\psi(w)|\underline{u} |+2(\xi-w)
                                   \quad \textrm {and} \quad
    \frac{\partial ^{2}H}{\partial w^{2}}=-\psi'(w)|\underline{u} |-2.
  \end{align*}  
  By hypothesis, the function $\Psi$ is odd, and hence so is its
  second derivative $\psi'$.  Moreover, $\psi'$ is positive in
  $(-\delta,0)$ and consequently negative in $(0,\delta)$.  By
  assumption (M4), there is a unique $\iota \in(-\delta,0)$ such that
  $\psi'$ is strictly increasing on $[-\delta, \iota]$ and strictly
  decreasing on $[\iota ,0]$, and so $\psi'$ is  strictly
  decreasing on $[0,-\iota]$ and strictly increasing on
  $[\iota ,\delta]$. Therefore, the function
  $\partial ^{2}H/\partial w^{2}$ has at most two zeros in
  $[-\delta,\delta]$, both in the interval $(0,\delta )$. It now
  follows by a sign change argument that the first partial
  $\partial H/\partial w$ has at most three zeros in
  $[-\delta,\delta]$; these must be separated by zeros of $\partial
  ^{2}H/\partial w^{2}$.

  According to our standing assumptions, the reward function $\Psi$
  is identically $1$ in the interval $[\delta,\infty )$. Hence, since
  $\xi>\delta$, the function $w\mapsto H(\xi,w)$ is \emph{negative} in the
  interval $[\delta,\xi )$. Define $w_{*}\in (-\delta,\delta)$ by
  \begin{displaymath}
    w_{*}:=\max \xset {w \in [-\delta,\delta] \,:\, H(\xi,w)=0};
  \end{displaymath}
  Lemma~\ref{lemma:unique-Opt} ensures that $w_{*}$ is well-defined,
  and that $w_{*}$ is a critical point. To complete the proof, we must
  show that there are no other points $w \not = w_{*}$ where $H(\xi,w)=0$.

  Since $H(\xi,w)=0$ at the endpoints $w=w_{*}$ and $w=\xi$ and
  $H(\xi,w)<0$ for $w_{*}<w<\xi$, it follows that the function
  $w\mapsto H(\xi, w)$ attains a minimum value at some point
  $w_{+}\in (w_{*},\xi)$; this must also be a critical point. This
  accounts for two of the (at most) three critical points. Now suppose
  that there were a second point $w_{**}<w_{*}$ in the interval
  $(-\delta,\delta)$ where $H(\xi,w_{**})=0$. Since this point $w_{**}$ would
  be a local maximum of $w\mapsto H(\xi,w)$, it would necessarily be a
  third critical point. But because $H(\xi,w)<0$ for all
  $w\in (w_{**},w_{*})$, there would be at least one local minimum of
  $w\mapsto H(\xi,w)$ in the interval $ (w_{**},w_{*})$; this would be
  a fourth critical point, contradicting the fact that there are at
  most three.

\end{proof}

\end{document}